\documentclass[10pt,journal,onecolumn,compsoc]{IEEEtran}

\usepackage{balance}

\usepackage[nocompress]{cite}

\usepackage{balance}
\usepackage[show]{chato-notes}

\usepackage{rotating}
\usepackage{graphicx}
\usepackage{tikz}
\usepackage{color}
\usepackage{cite}
\usepackage{url}
\usepackage{booktabs}

\usepackage{amsmath,amssymb, amsthm}
\usepackage{mathtools}
\usepackage{dsfont}
\usepackage{epsfig}
\usepackage{framed}
\usepackage[ruled,lined,linesnumbered,noend]{algorithm2e}
\usepackage{multirow}
\usepackage{capt-of}
\usepackage{textcomp}
\usepackage{tabularx}
\usepackage{bbm}
\usepackage{times}
\usepackage{hyperref}

\hypersetup{
    bookmarks=true,         
    unicode=false,          
    pdftoolbar=true,        
    pdfmenubar=true,        
    pdffitwindow=false,     
    pdfstartview={FitH},    
    pdftitle={Maximizing the Diversity of Exposure in a Social Network},    
    pdfauthor={A.~Matakos et al.},     
    pdfnewwindow=true,      
    colorlinks=false,       
    linkcolor=black,          
    citecolor=black,        
    filecolor=black,      
    urlcolor=black           
}

\newcommand{\eat}[1]{}

\newtheorem{definition}{Definition}

\newtheorem{theorem}{Theorem}
\newtheorem{corollary}{Corollary}

\newtheorem{lemma}{Lemma}
\newtheorem{problem}{Problem}

\DeclarePairedDelimiter\abs{\lvert}{\rvert}

\DeclareMathOperator*{\argmax}{arg\,max}

\newcommand{\NP}{\ensuremath{\mathbf{NP}}\xspace}
\newcommand{\NPhard}{\NP-hard\xspace}

\newcommand{\SP}{\ensuremath{\#\mathbf{P}}\xspace}

\newcommand{\SPhard}{\SP-hard\xspace}

\newcommand{\spara}[1]{\smallskip\noindent{\bf #1}}

\newcommand{\squishlist}{
 \begin{list}{$\bullet$}
  {  \setlength{\itemsep}{0pt}
     \setlength{\parsep}{3pt}
     \setlength{\topsep}{3pt}
     \setlength{\partopsep}{0pt}
     \setlength{\leftmargin}{2em}
     \setlength{\labelwidth}{1.5em}
     \setlength{\labelsep}{0.5em}
} }
\newcommand{\squishlisttight}{
 \begin{list}{$\bullet$}
  { \setlength{\itemsep}{0pt}
    \setlength{\parsep}{0pt}
    \setlength{\topsep}{0pt}
    \setlength{\partopsep}{0pt}
    \setlength{\leftmargin}{2em}
    \setlength{\labelwidth}{1.5em}
    \setlength{\labelsep}{0.5em}
} }

\newcommand{\squishdesc}{
 \begin{list}{}
  {  \setlength{\itemsep}{0pt}
     \setlength{\parsep}{3pt}
     \setlength{\topsep}{3pt}
     \setlength{\partopsep}{0pt}
     \setlength{\leftmargin}{1em}
     \setlength{\labelwidth}{1.5em}
     \setlength{\labelsep}{0.5em}
} }

\newcommand{\squishend}{
  \end{list}
}

\newcommand{\naturals}{\ensuremath{\mathbb{N}}\xspace}
\newcommand{\calE}{\ensuremath{\mathcal{E}}\xspace}
\newcommand{\calF}{\ensuremath{\mathcal{F}}\xspace}
\newcommand{\calM}{\ensuremath{\mathcal{M}}\xspace}
\newcommand{\probG}{\ensuremath{\mathcal{G}}\xspace}
\newcommand{\multiG}{\ensuremath{\widetilde{G}}\xspace}
\newcommand{\multiE}{\ensuremath{\widetilde{E}}\xspace}
\newcommand{\multiP}{\ensuremath{\widetilde{p}}\xspace}
\newcommand{\multiR}{\ensuremath{\widetilde{R}}\xspace}
\newcommand{\sampleR}{\ensuremath{\mathcal{\multiR}}\xspace}
\newcommand{\prob}{\ensuremath{\mathrm{Pr}}\xspace}
\newcommand{\Exp}{\ensuremath{\mathbb{E}}\xspace}
\newcommand{\Var}{\ensuremath{\mathrm{Var}}\xspace}
\newcommand{\p}{\ensuremath{\mathrm{path}}\xspace}

\newcommand{\bigO}{\ensuremath{\mathcal{O}}\xspace}
\newcommand{\xdivscore}[1]{{\Exp}\left[F(#1)\right]} 
\newcommand{\xgaindivscore}[2]{{\Exp}\left[F(#1 \mid #2)\right]} 
\newcommand{\greedyA}{\ensuremath{A_{G}}\xspace} 
\newcommand{\optA}{\ensuremath{{A^{*}}}\xspace} 
\newcommand{\optvalue}{\ensuremath{{\mathrm{OPT}}}\xspace} 
\newcommand{\approxgrA}{\ensuremath{\tilde{A}}\xspace} 
\newcommand{\LB}{\ensuremath{\mathrm{LB}}\xspace}
\newcommand{\LBnaive}{\ensuremath{\mathrm{LB}_0}\xspace}

\newcommand{\rcGreedy}{\ensuremath{\mathrm{RC\text{-}Greedy}}\xspace}
\newcommand{\fastAlgo}{\textsc{\large tdem}\xspace}
\newcommand{\fastAlgoLong}{{Two-phase Di\-ver\-sity Ex\-po\-sure Ma\-xi\-mi\-za\-tion}\xspace}
\newcommand{\tim}{\textsc{\large tim}\xspace}
\newcommand{\imm}{\textsc{\large imm}\xspace}
\newcommand{\rr}{\textsc{\large rr}\xspace}
\newcommand{\rc}{\textsc{\large rc}\xspace}

\newcommand{\algDClose}{\textsc{Min-Var}\xspace}
\newcommand{\algDFar}{\textsc{Max-Var}\xspace}
\newcommand{\algDWeighted}{\textsc{Myopic}\xspace}

\newcommand{\algTdem}{\textsc{\large tdem}\xspace}

\newcommand{\dataset}[1]{\texttt{#1}\xspace}
\newcommand{\dtstGroup}[1]{\dataset{#1}}
\newcommand{\dtstSub}[2]{\dataset{#1:#2}}

\newcommand{\dblp}{DBLP}
\newcommand{\nips}{G}
\newcommand{\twitter}{Twitt}
\newcommand{\uste}{Tweet}
\newcommand{\usep}{TPair}

\newcommand{\dtstClDBLP}{\dtstGroup{\dblp}}
\newcommand{\dtstDBLPbs}{\dtstSub{\dblp}{BSch}} 
\newcommand{\dtstDBLPpy}{\dtstSub{\dblp}{PYu}} 
\newcommand{\dtstDBLPcp}{\dtstSub{\dblp}{CPap}} 

\newcommand{\dtstClNIPS}{\dtstGroup{\nips}}
\newcommand{\dtstObamacare}{\dtstSub{\nips}{ObamaC}} 
\newcommand{\dtstPhone}{\dtstSub{\nips}{IPhone}} 
\newcommand{\dtstUselect}{\dtstSub{\nips}{US-elect}} 
\newcommand{\dtstBrexit}{\dtstSub{\nips}{Brexit}} 
\newcommand{\dtstFrack}{\dtstSub{\nips}{Fracking}} 
\newcommand{\dtstAbort}{\dtstSub{\nips}{Abortion}} 

\newcommand{\dtstClTWT}{\dtstGroup{\twitter}}
\newcommand{\dtstTWTXL}{\dtstSub{\twitter}{XL}} 
\newcommand{\dtstTWTS}{\dtstSub{\twitter}{Follow}} 

\newcommand{\dtstClTWTE}{\dtstGroup{\uste}}
\newcommand{\dtstTWTEMFive}{\dtstSub{\uste}{M5}} 
\newcommand{\dtstTWTESFive}{\dtstSub{\uste}{S5}} 
\newcommand{\dtstTWTESTwo}{\dtstSub{\uste}{S2}} 

\newcommand{\dtstClTWEP}{\dtstGroup{\usep}}
\newcommand{\dtstTWEPa}{\dtstSub{\usep}{Z}} 
\newcommand{\dtstTWEPb}{\dtstSub{\usep}{Y}} 
\newcommand{\dtstTWEPc}{\dtstSub{\usep}{X}} 

\DeclareMathOperator{\avg}{avg}
\newcommand{\nbN}{\ensuremath{n}\xspace}
\newcommand{\nbE}{\ensuremath{m}\xspace}
\newcommand{\density}{\ensuremath{d(G)}\xspace}

\newcommand{\score}{\ensuremath{F(A)}}

\usepackage{expl3, xparse, siunitx}
\ExplSyntaxOn
\NewDocumentCommand { \calcnum } { O{} m }
{ \num [  round-mode=places , round-precision=3 , group-separator={,}, group-minimum-digits=4, #1] { \fp_to_decimal:n {#2} } }
\ExplSyntaxOff

\graphicspath{{figs/}}

\title{Maximizing the Diversity of Exposure\\ in a Social Network}

\author{Antonis~Matakos, Cigdem~Aslay, Esther~Galbrun, and Aristides~Gionis%
\IEEEcompsocitemizethanks{
\IEEEcompsocthanksitem A.~Matakos is with the Department of Computer Science, Aalto University, Finland.\protect\\
E-mail: firstname.lastname@aalto.fi%
\IEEEcompsocthanksitem C.~Aslay is with Department of Computer Science, Aarhus University, Denmark.\protect\\
E-mail: cigdem@cs.au.dk%
\IEEEcompsocthanksitem E.~Galbrun is with the School of Computing, University of Eastern Finland.\protect\\
E-mail: esther.galbrun@uef.fi%
\IEEEcompsocthanksitem A.~Gionis is with the Department of Computer Science, KTH Royal Institute of Technology, Sweden, 
and the Department of Computer Science, Aalto University, Finland.\protect\\
E-mail: argioni@kth.se\protect
}
\thanks{Part of this work was done while the authors were with Aalto University, in Finland. This research is supported by the Academy of Finland projects AIDA (317085) and MLDB (325117),
the ERC Advanced Grant REBOUND (834862), 
the EC H2020 RIA project SoBigData (871042), and 
the Wallenberg AI, Autonomous Systems and Software Program (WASP) 
funded by the Knut and Alice Wallenberg Foundation.
\protect}}

\begin{document}


\IEEEtitleabstractindextext{%
\begin{abstract}
Social-media platforms have created new ways for citizens to stay informed and participate in public debates.
However, to enable a healthy environment for information sharing, social deliberation, and opinion formation, citizens need to be exposed to sufficiently diverse viewpoints that challenge their assumptions, instead of being trapped inside filter bubbles.
In this paper, we take a step in this direction and propose a novel approach to maximize the diversity of exposure in a social network.
We formulate the problem in the context of information prop\-a\-ga\-tion,
as a task of recommending a small number of news articles to selected users.
In the proposed setting, we take into account content and user leanings,
and the probability of further sharing an article.
Our model allows to capture the balance between max\-i\-miz\-ing the spread of information and ensuring the exposure of users to diverse viewpoints.

The resulting problem can be cast as maximizing a monotone and submodular function, subject to a matroid constraint on the allocation of articles to users. It is a challenging generalization of the influence-maximization problem. Yet, we are able to devise scalable approximation algorithms by introducing a novel extension to the notion of random reverse-reachable sets. We experimentally demonstrate the efficiency and scalability of our algorithm on several real-world datasets.
\end{abstract}

}

\maketitle

\section{Introduction}
\label{sec:intro}
Over the past decade, the emergence of social-media platforms has changed society in unprecedented ways, completely altering the landscape of societal debates and creating radically new ways of collective action. In this networked public sphere, members of society have access to a public podium where they can participate in public debate and speak up about topics they deem to be of public concern. This emerging environment of participatory culture has made the diversity of citizens' views more relevant than ever before.

While having the potential to expose individuals to diverse opinions, social-media platforms typically resort to personalization algorithms that filter content based on social connections and previously expressed opinions, creating filter bubbles~\cite{pariser11filter}.
The resulting echo chambers tend to amplify and reinforce pre-existing opinions, catalyzing an environment that has a corrosive effect on the democratic debate.

In this paper we propose \emph{a novel approach towards breaking filter bubbles}.  We consider social-media discussions around a topic that are characterized by a number of viewpoints falling within a predefined spectrum of opinions. To accurately model the dynamics of social-media platforms, we assume that each viewpoint is represented by a number of items (articles, posts) propagating through the network, via messages, re-shares, retweets, etc. Furthermore, we assume that each individual is associated with a \textit{leaning} with respect to the issue, which impacts whether they will further disseminate any article they come across, depending on how it aligns with their leaning. We think that this is a realistic assumption, since, for example, an individual with conservative leaning will be reluctant to share an article with liberal leaning.

We refer to the diversity of the information that a user is exposed to as the user's ``\emph{diversity exposure level}''. It depends on the viewpoint expressed in the articles the user consumes, referred to as \emph{article leanings}, and the users' existing viewpoint on the matter, referred to as \emph{user leanings}. We assume that the diversity exposure level of users can be increased through \emph{content recommendations} made by the social-media platform. Considering that filter bubbles result from a lack of exposure to diverse viewpoints, our aim is to measure and maximize the total diversity exposure levels of all users in the network. 

Our problem can be naturally defined in an \emph{information-pro\-pa\-ga\-tion setting}~\cite{kempe03}: we ask to select a small number of seed users and the articles that should be recommended to them so as to maximize the total diversity of exposure in the network. Since the recommended articles are inserted into the timeline of the users, disrupting the organic flow of the content in the network, we also consider a limit on the number of articles that can be recommended to a user in this way.  

An attractive aspect of our problem setting is that it conso\-li\-da\-tes many aspects of the functionality of real-life social net\-works. By incorporating article leanings, user leanings, and the probabilities of further sharing an article, we ask to find the recommendations that translate to a good spread and simul\-taneously maximize the diversity exposure level of the users. To better understand the interplay between spread and diversity, observe that assigning articles that match the users' predisposition is likely to result in a high spread but minimal increase of diversity, while recommending articles that are opposed to users' predispositions, will likely result in high diversity locally but hinder the spread of the articles. This trade-off is central to the diversity-maximization problem we consider.

We show that taking all the aforementioned components into account, the problem of maximizing the diversity of exposure in a social network can be cast as maximizing a monotone and submodular function subject to a matroid constraint on the allocation of articles to users. We show that this problem is \NP-hard and is far more challenging than the classical influence-maximization problem. 
We introduce a non-trivial generalization of random reverse-reachable sets ({\rr}-sets)~\cite{borgs14}, which we call random \emph{reverse co-exposure} sets ({\rc}-sets), for accurately estimating the diversity of exposure in a social network. We propose a scalable approximation algorithm, named \emph{\fastAlgoLong} (\fastAlgo), that leverages random {\rc}-sets and an adaptive sample size determination procedure, ensuring  quality guarantee on the returned solution with high probability. 

Although our approach belongs to a large body of work on information propagation and breaking filter bubbles, there are significant differences and novelties. In particular:

\smallskip
\begin{itemize}
\item
We are the first to address the problem of maximizing the diversity of exposure and breaking filter bubbles in an item-aware information propagation setting. We leverage several real-world aspects of social-media functionality, such as how users consume and share articles, while considering user-article dependent propagation probabilities.
\item 
We formally define the problem of maximizing the diversity of exposure, prove its hardness, and develop a simple greedy algorithm. 
\item
We then introduce the notion of random {\em reverse co-exposure sets} and devise a scalable instantiation of the greedy algorithm with provable guarantees. 
\item
Our extensive experimentation on real-world datasets confirms that our algorithm is scalable and delivers high quality solutions, significantly outperforming several natural baselines. 
\end{itemize}

A preliminary version of our work provided a first theoretical and experimental treatment of the problem under a simpler formulation~\cite{aslay2018maximizing}. Specifically, we previously defined the diversity exposure level of a user to be equal to the breadth of leanings spanned by the items the user is exposed to, in addition to the user's own leaning.  In this paper, we extend our preliminary results in several directions. First, we propose a refined scheme to quantify the diversity exposure level of a user. The new diversity definition measures not only the range of leanings in a set of items but also their spread within this range. That is, our refined scores does not only look at the extremes of represented leanings but also at how well intermediate leanings are covered. Second, we show that the total diversity exposure function remains submodular and monotone and we extend our scalable approximation framework based on random reverse co-exposure sets to operate under this new score. Finally, we provide additional experiments on many real-world datasets.

\section{Related work}
\label{sec:related}

Our work relates to the emerging line of research on breaking filter bubbles in social media. To the best of our knowledge, this is the first work to approach this problem from the angle of maximizing the diversity of information exposure in an item-aware independent-cascade model.

\smallskip
\noindent
\textbf{Filter bubbles and echo chambers.}
Recently, there have been a number of studies on the effects of ``echo chambers''~\cite{bakshy15exposure,garrett09echo}, where users are only exposed to information from like-minded individuals, and of ``filter bubbles''~\cite{bakshy15exposure,pariser11filter}, where algorithms only present personalized content that agrees with the user's viewpoint. In particular, Garrett et al.~\cite{garrett09echo} observed that news stories containing opinion-challenging informations spread less than other news.

In order to measure how strongly these phenomena manifest themselves on social media, a significant body of work has emerged that focuses on measures for characterizing polarization~\cite{Matakos17,akoglu14,conover11political,garimella16quantifying,guerra13measure,xi2018quantifying}.

In a similar vein to ours, previous works have studied the problem of diversifying exposure. This task presents various aspects, such as the questions of \emph{who} to target, \emph{what} viewpoints to promote, or \emph{how} best to present possibly opposing viewpoints to users~\cite{Liao2014}.
Recent approaches focus on targeting users so as to reduce the polarization of opinions and bridge opposing views~\cite{Matakos17, Musco17,garimella17reducing, xi2018quantifying}. These works consider an \emph{opinion-formation} model whereas our underlying model is an \emph{influence-propagation} model. 
From this angle, the works by Garimella et al.~\cite{garimella17balancing} and Rawal and Khan 
\cite{rawal2019maximizing} are closest to our work. They consider an influence propagation setting, where two conflicting campaigns propagate in the network and aim to maximize the number of users exposed to both campaigns. However, the granularity of our setting is finer since we consider items with leanings lying across a spectrum rather than two opposing sides. Additionally, we consider the leanings of users, which affect the propagation probabilities. Since our goal is to identify assignments of items to users, we aim to identify both the users to target and the viewpoints to expose them to. 

\smallskip
\noindent
\textbf{Influence maximization.}
Our problem is also related to the work on influence maximization.
Kempe et al.~\cite{kempe03} formalized the influence maximization problem and proposed two propagation models, the \emph{independent-cascade model} and the \emph{linear-threshold model}. These models were subsequently extended to handle the case of multiple competing campaigns in a network~\cite{Bharathi2007competitive,Borodin2010,Valera15}.
As other authors have suggested, we consider a \emph{central authority} selecting the seed set~\cite{Borodin2017, Aslay2015Viral, Aslay2017revenue, garimella17balancing}. Our setting is related to social advertising~\cite{Aslay2015Viral, Aslay2017revenue}, which also considers item-aware propagation models, aiming to allocate ads so as to maximize the engagement of users.
Key to our work is the idea of \emph{reverse reachable sets} introduced by~\cite{borgs14}, which provides scalable solutions for the influence maximization problem. Subsequent works~\cite{Cohen2014sketch,NguyenTD16,tang2015influence,tang14} introduced techniques to improve upon this idea even further. We extend these ideas to our setting, and obtain an algorithm that scales to very large datasets.

\section{Problem definition}
\label{sec:problem}
\spara{Notation.}
The input to the problem of diversifying exposure consists of the following ingredients:
(i) a directed social graph $G=(V,E)$,
with $\lvert V \rvert = n$ nodes and $\lvert E \rvert = m$ edges, where nodes represent users and a directed edge $(u,v)$ indicates that user $v$ follows user $u$,
thus, $v$ can see and propagate posts by $u$;
(ii) a set $H$ of (news) items on a (possibly controversial) topic, with $\lvert H \rvert = h$;
(iii) item-specific propagation probabilities $p^i_{uv}$,
for all items $i \in H$ and edges $(u,v) \in E$,
where $p^i_{uv}$ represents the probability that item~$i$ will propagate from user $u$ to user $v$;
(iv) a \emph{leaning function} $\ell: V \cup H \rightarrow [-1,1]$
that quantifies the polarity of the viewpoints of items and users with respect to the considered issue or topic.

\spara{Cascade model.} We assume that the propagation of an item $i \in H$
from user $u$ to user $v$ follows the \emph{independent-cascade model}
with parameter $p^i_{uv}$, and is independent from the propagation of other items to $v$ from its in-neighbors.
Thus, once $u$ becomes active on item~$i$ at time $t$,
the probability $p^i_{uv}$ that $u$ succeeds in activating $v$ with item $i$ at time $t+1$ is independent 
of other items with which user $u$ or other in-neighbors of $v$ might succeeded to activate $v$ at any time. 
We incorporate the different tendency of users to share items
with leanings diverging from or similar to their own by allowing item-specific propagation probabilities for each edge.
Hence, $p^i_{uv}$ implicitly takes into account the leanings of users $u$ and $v$ and of item~$i$.
The leaning of a user reflects the user's viewpoint, which is considered to be stable. 
Therefore, we assume that the propagation probabilities remain fixed over time, and probabilities $p^i_{uv}$ are constant values input to our cascade model. We consider the estimation of user leanings and propagation probabilities as orthogonal to our work.

\spara{Quantifying diversity of exposure.}
We say that user $v$ is \emph{exposed} to item~$i$ if $v$ is \emph{activated} on item~$i$,
either by an in-neighbor that is active on item~$i$, or due to being a \emph{seed node} for item~$i$.
Consider a user $v$ that is exposed to a set $I\subseteq H$ of items.
It follows that user $v$ is exposed to a set of leanings $\{\ell(v)\} \cup \{\ell(i) : i \in I\}$. Intuitively, we want each user to be aware of a multitude of viewpoints, while also retaining a balanced perspective. To account for 
both factors, we define a \emph{penalty} function that quantifies the lack of diversity of exposure. 

Specifically, we want to penalize large gaps in the spread of leanings, which correspond to ranges of opinions not represented among items the user is exposed to. Therefore, the function is defined for each user by considering the set of distinct leanings he is exposed to, sorted by polarity, and taking the sum of squared distances between consecutive leanings, also accounting for the extreme values of leanings. We consider that each item contributes only once to the diversity of exposure of a user. Therefore seeing the same article multiple times should have no impact on the objective.

We let $L(v,I) = \langle \ell_1, \dots, \ell_{\eta} \rangle$ denote the set $\{\ell(v)\} \cup \{\ell(i) : i \in I\}  \cup \{-1, 1\}$ sorted by increasing values, i.e., such that $\ell_i \leq \ell_j$ for all $i < j$.
This set contains the distinct leanings among the items in $I$ that user $v$ has been exposed to, as well as the two extreme leanings across the spectrum of opinions, $\ell_1=-1$ and $\ell_{\eta}=1$.
Then, we define the penalty for node $v$, $g_v: 2^{H} \rightarrow [0,4]$, as 
\begin{align}
	g_v(I) =\sum_{j=1}^{\eta-1} (\ell_{j+1} - \ell_{j} )^2, \quad \ell \in L(v,I).
\end{align}
Given the penalty function $g_v(I)$ that quantifies the lack of diversity in the leanings of the items $I$ that $v$ is exposed to, we define the \emph{level of diversity exposure} $f_v: 2^{H} \rightarrow [0,1]$ of $v$ as
\begin{align}
	f_v(I) = 1-\frac 14 g_v(I).
\end{align}
Notice that the range of the diversity exposure function $f_v$ is $[0,1]$, where a value of $1$ corresponds to the maximum possible diversity of exposure. 

To motivate the definition of our diversity function $f_v$ we provide the following two lemmas, 
which illustrate some of its desirable properties. 

\begin{lemma}
\label{lemma:monotonicity}
For all $I\subseteq J \subseteq H$, we have $f_v(I) \le f_v(J)$.
\end{lemma}

Lemma~\ref{lemma:monotonicity}, 
for which the proof is provided as part of Lemma~\ref{lemma:submodularity},
states that $f_v$ is {\em monotone}, i.e., the diversity exposure level of $v$ cannot decrease as the user is exposed to more items. 

Next we formally show that, if we fix the number of items that user $v$ will see, then the configuration in which $f_v$ is maximized corresponds to the \emph{desired} scenario where the user leaning of $v$ and the leanings of the items $v$ is exposed to are equally spaced across $[-1,1]$.

\begin{lemma}
\label{lemma:uniformity}
Consider a set of items $I$, so that $I$ has fixed cardinality $\kappa$. 
Then, the diversity function $f_v(I)$ is maximized if 
the leanings of the items in $I$ are equidistantly positioned in the interval $[-1,1]$.
\end{lemma}

\begin{proof}
For the sake of simplicity, and without loss of generality, assume that neither the leaning of $v$ nor the extreme leanings $-1$ and $1$ are represented in $I$, so that $|L(v,I)| = \kappa+3$. Let $r_j= \ell_{j+1}-\ell_j$, for $j = 1,\ldots,\kappa+2$. 
Notice that $\sum_{j=1}^{\kappa+2} r_j$ is a constant that depends only on how we model the range of the leanings, i.e., for $[-1,1]$, we have $\sum_{j=1}^{\kappa+2} r_j = 2$. Remember that by definition $f_v(I) $ is maximized  whenever $g_v(I)$ is minimized. Then, solving the equations resulting from $g_v'(I) = 0$ and $\sum_{j=1}^{\kappa+2} r_j = 2$, we see that $g_v(I)$ attends its minimum value when \[r_1=\ldots=r_{\kappa+2}=\frac{2}{\kappa+2}\;.\]
\end{proof}

\spara{Assignment to seed nodes.}
We consider selecting a set of users in $V$ as the \emph{seed nodes} 
and expose them to a subset of items from $H$.
Let $\calE = V \times H$ denote the set of all possible (user, item)
pairs and let $A \subseteq \mathcal{E}$ denote an assignment such that the set $A_i = \{u\in V: (u,i) \in A \}$ contains the seed nodes selected for initial exposure to item~$i$ and the set $A_u = \{i\in H: (u,i) \in A \}$ contains the items assigned to seed node $u$.
For each $v \in V$, we denote by $I_v(A)$ the set of items that $v$ is exposed to when the propagation process started from assignment $A$ converges. The \emph{diversity of exposure score} $F(A)$ of an assignment $A$
is then defined as the sum of diversity exposure levels
of all the users resulting from the assignment $A$ in $G$
\begin{align}
F(A) = \sum_{v \in V} f_v\big(I_v(A)\big).
\end{align}
Note that the function 
$f_v(I_v): 2^{\mathcal{E}} \rightarrow [0,1]$
is a composition
$f_v(I_v) = f_v \circ I_v$
of the functions
$I_v: 2^{\mathcal{E}} \rightarrow 2^{H}$ and
$f_v:  2^{H} \rightarrow [0,1]$.
We will later use this fact to show that $f_v(I_v)$
is a submodular function over \calE.

\spara{Constraints on assignments.}
We assume that we are interested in assignments of size at most $k \in \naturals$.
Moreover, taking into account the limited attention bound of users,
which can be user-specific~\cite{linHWY14},
we also limit the number of items that a user can be seeded with.\!\footnote{As in previous work \cite{Aslay2015Viral}, 
we do not assume any attention bound on the number of items that are not recommendations in our problem definition,
i.e., items that appear in the news-feed of the users in the social network,
as such items are part of the organic operation of the network.} 
We model this using an \emph{attention bound constraint}
$k_u \in \naturals$ for each user $u \in V$.
We say that an assignment $A$ is feasible if
$\lvert A \rvert \le k$ and $\lvert A_u \rvert \le k_u$, for each seed node $u$.

\spara{Assumptions.}
We assume that there exist $e,e' \in V \cup H$
such that $\ell(e) \neq \ell(e')$.
This weak assumption is simply a bare minimum requirement
on the diversity of the leanings of the users and items,
aligned with the motivation of the problem.
We will use this assumption in the greedy approximation analysis
to constrain the optimal values of expected diversity exposure score to $\mathbb{R}_{+}$.

We are now ready to formally define our problem.

\begin{problem}[Diversity Exposure Maximization]
\label{pr:diversityMax}
Given a directed social graph $G = (V,E)$ with user leanings $\ell(v)$, for all $v \in V$,
a set of items $H$ with item leanings $\ell(i)$, for all $i \in H$,
item-specific propagation probabilities $p^i_{uv}$, for all $(u,v) \in E$ and all $i \in H$,
positive integers $k_u$ for the attention bound constraints, for all $u \in V$, and
a positive integer $k$, find a feasible assignment $A$
that maximizes the expected diversity exposure score
\begin{equation*}
\begin{aligned}
& \underset{A \subseteq \mathcal{E}}{\text{\em maximize}}
& & \xdivscore{A} \\
& \text{\em subject to}
& & \lvert A \rvert \le k, \\
&&& \lvert A_u \rvert \le k_u, \,\text{\em  for all } u \in V.
\end{aligned}
\end{equation*}
\end{problem}
We use $\optA$ to denote the optimal solution of Problem~\ref{pr:diversityMax}, and
$\optvalue = \xdivscore{\optA}$ to denote its expected score in $G$.

\section{Theoretical analysis}
\label{sec:theory}
\subsection{Possible-world semantics}
{A \emph{probabilistic graph} $\probG = (V,E,p)$, comprises a vertex set $V$ and an edge set $E$, where each edge $e$ is associated with a probability $p_e\in p$.} Given a probabilistic graph,
a \emph{possible world} is a \emph{deterministic} graph obtained from $\probG$
with edges sampled independently according to $p$.
We now introduce the \emph{possible-world model} for our problem
that can capture the co-exposure of nodes to items resulting from any given assignment.

We start by defining a \emph{directed edge-colored multigraph}
$\multiG = (V, \multiE, \multiP)$ from $G = (V,E)$,
by creating $h$ copies of each directed edge $(u,v) \in E$.
For each item $i \in H$ we create a parallel edge $(u,v)_i$ in $\multiG$,
having distinct color and associated probability $p^i_{uv}$.
We interpret $\multiG$ as a probability distribution over all subgraphs of $(V, \multiE)$,
i.e., we sample each edge $(u,v)_i \in \multiE$ independently at random with probability $p_{uv}^i$.
The probability of a possible world $g \sqsubseteq \multiG$ is given by
\begin{align}
\prob[g] = \prod_{i \in H} \prod_{(u,v)_i \in g} p^i_{uv} \prod_{(u,v)_i \in \multiE \setminus g} (1 - p^i_{uv}).
\end{align}

Let $\p_g^i(u,v)$ denote an indicator variable that equals $1$
if node $v \in V$ is reachable by node $u$ via the colored edges of $i$ in $g$,
and $0$ otherwise.
We say that a pair $(u,i)$ can \emph{color-reach} node $v$
if $\p_g^i(u,v) = 1$.
For an assignment $A$ and a node $v\in V$ let $I_v^g(A)$ be the set of items
that $v$ is exposed to, due to $A$, in network $g$.
It can be written as
\begin{align*}
I_v^g(A) = \{i \in H \mid \text{ exists } (u,i) \in A \text{ and } \p_g^i(u,v) = 1\}.
\end{align*}
The value of the objective $\xdivscore{A}$
in Problem~\ref{pr:diversityMax} is given by
\begin{eqnarray}
\label{eq:objPossibWorld}
\xdivscore{A} &=& {\Exp}\left[\sum_{v \in V} f_v(I_v^g(A)) \right] \nonumber\\
&=& \sum_{g \sqsubseteq \multiG} \prob[g] \sum_{v \in V} f_v(I_v^g(A)).
\end{eqnarray}

\subsection{Hardness and approximation}

We will first show that the objective function of  Problem~\ref{pr:diversityMax} is monotone and submodular.

\begin{lemma}
\label{lemma:submodularity}
The function
$\xdivscore{\cdot}$ is monotone and submodular.
\end{lemma}

\begin{proof}
To prove the lemma, we utilize the possible-world semantics.
It is well known that a non-negative linear combination of submodular functions is also submodular.
Therefore, to prove submodularity of $\xdivscore{\cdot}$,
it is sufficient to show that in any possible world $g \sqsubseteq \multiG$,
$f_v: 2^{\calE} \rightarrow [0,1]$ is submodular.
Similarly, to prove monotonicity of $\xdivscore{\cdot}$, it suffices to show the monotonicity of $f_v(\cdot)$ in any possible world $g$. 

Now, recall that we have $f_v(I_v^g(A)) = 1-\frac{1}{4}g_v(I_v^g(A))$. We will show that $g_v(I_v^g(A))$ is supermodular and monotonically non-increasing in $A$ which will directly imply the submodularity and monotonicity of  $f_v(I_v^g(A))$. 

First we show that $g_v(I_v^g(A))$ is monotonically non-increasing in $A$ by showing that
$g_v(I_v^g(A)) \ge g_v(I_v^g(A\cup e))$ for any $A \subseteq \calE$ and $(w,x) \in \calE \setminus A$. 

First, consider the case $path_g^x(w,v) = 0$. Notice that in this case we have $g_v(I_v^g(A)) = g_v(I_v^g(A\cup \{(w,x)\}))$ as $I_v^g(A) = I_v^g(A\cup \{(w,x)\})$. Now, consider the case $path_g^x(w,v) = 1$. In this case, we have $I_v^g(A \cup \{(w,x)\}) = I_v^g(A) \cup \{x\}$. Let $i, j \in I_v^g(A)$ be such that $\ell(i)$ and $\ell(j)$ are the immediate predecessor and successor of $\ell(x)$ in $L(v,I_v^g(A \cup \{(w,x)\}))$ respectively, i.e., $\not\exists y \in L(v,I_v^g(A \cup \{(w,x)\}))$ such that $\ell(i)\le \ell(y) \le \ell(x)$ or $\ell(x) \le \ell(y) \le \ell(j)$. 

Then we have, 
\begin{align*}
g_v (I_v^g  & (A \cup \{(w,x)\})) - g_v(I_v^g(A)) \\ 
& =  (\ell(i) -\ell(x))^2+(\ell(x)-\ell(j))^2 - (\ell(i)-\ell(j))^2 \\
& = (\ell(i)-\ell(x))^2+(\ell(x)-\ell(j))^2 \\ 
& \quad\,\, - (\ell(i)-\ell(x)+\ell(x)-\ell(j))^2 \\
& \le 0. 
\end{align*}
We have just shown that $g_v(I_v^g(A))$ is monotonically non-increasing in $A$.

We now show that $g_v(I_v^g(A))$ is supermodular in $A$.
Let $g_v(I_v^g((w,x) \mid A))$ denote the marginal decrease in the penalty  
when $(w,x)$ is added to the assignment $A$:
\[
g_v(I_v^g((w,x) \mid A)) =  g_v(I_v^g(A \cup \{(w,x)\})) - g_v(I_v^g(A)).
\]
To show that $g_v(I_v^g(\cdot))$ is supermodular, we need to show that
\[
g_v(I_v^g((w,x) \mid A)) \le g_v(I_v^g((w,x) \mid B)),
\]
for any $A \subseteq B \subseteq \calE$ and $(w,x) \not\in B$.

Let $B = A \cup \{(z,y)\}$ for some $(z,y) \in \calE \setminus A$. First, notice that if $\p^x_g(w,v) = 0$ and $\p^y_g(z,v) = 0$, then the analysis is trivial, since, $I_v^g((w,x) \mid A) = I_v^g((w,x) \mid B) = I_v^g(A)$, resulting in $g_v(I_v^g((w,x) \mid A)) = g_v(I_v^g((w,x) \mid B)) = 0$. Next, we provide the analysis for the case $\p^x_g(w,v) = 1$ and $\p^y_g(z,v) = 1$, and omit the analysis of the other two cases in which either $\p^x_g(w,v) = 0$ or $\p^y_g(z,v) = 0$ as their analysis use similar arguments.

We now start the analysis for the case $\p^x_g(w,v) = 1$ and $\p^y_g(z,v) = 1$. 
To do so, we perform case-by-case analysis based on how $\ell(x)$ is compares to the leanings in $L(v,I_v^g(A))$ and $L(v,I_v^g(B))$.

Let $i, j \in I_v^g(A)$ be such that $\ell(i)$ and $\ell(j)$ are the immediate predecessor and successor of $\ell(x)$ in $L(v,I_v^g(A \cup \{(w,x)\}))$. 
Next, we consider the following two cases.

\begin{itemize}
	\item {Case 1}: $y$ is such that $\ell(y) < \ell(i)$ or $\ell(y)>\ell(j)$. Then we have $g_v(I_v^g((w,x) \mid B)) = g_v(I_v^g((w,x) \mid A))$.

	\item {Case 2}: $y$ is such that $\ell(i)\le \ell(y)\le \ell(x)$.
	Then we have
\begin{align*}
g_v & (I_v^g((w,x) \mid A)) - g_v(I_v^g((w,x) \mid B)) \\
 & = (\ell(i)-\ell(x))^2-(\ell(i)-\ell(j))^2-(\ell(y)-\ell(x))^2 \\
 & \quad\,\, +(\ell(y)-\ell(j))^2 \\
 & = -2\ell(i)\ell(x)+2\ell(i)\ell(j)+2\ell(y)\ell(x)-2\ell(y)\ell(j) \\
 & \quad\,\, + (\ell(x)-\ell(j))(2\ell(y)-2\ell(i)) \le 0
\end{align*}
	\item {Case 3}: $y$ is such that $\ell(x)\le \ell(y)\le \ell(j)$. 
	This case is symmetric to Case 2, so we omit the proof for brevity.
\end{itemize}
 
\end{proof}

\begin{theorem}
\label{theo:hardness}
Problem~\ref{pr:diversityMax} is \NPhard.
\end{theorem}

\begin{proof}
We will show that Problem~\ref{pr:diversityMax}
contains the influence maximization problem as a restricted special case,
which is shown to be \NPhard~\cite{kempe03}.
Consider the case where $H$ consists of a single item, which we denote by $i'$.
Let $\ell(i') = 0$ and $\ell(v) = 1$, for all $v \in V$.
Notice that since $h = 1$, the multiplicity of each edge in $\multiG$ is $1$.
In any $g \sqsubseteq \multiG$, if a node $v$ is exposed to $A$
then $f_v(I_v^g(A)) = \frac{1}{2}$,
while  $f_v(I_v^g(A)) = 0$ if $v$ is not exposed to $A$.
For any assignment $A$, define the seed set $S = \{u \mid (u,i') \in A\}$.
Let $C_g(S)$ denote the number of nodes reachable by $S$ in $g$.
Notice the equivalence $C_g(S) = 2\sum_{v \in V} f_v(I_v^g(A))$ in $g$.
Now, let $S^{*} \subseteq V$ denote the optimal solution
to the influence maximization problem with parameter $k$.
The expected spread of $S^{*}$ is given by
${\Exp}\left[C_g(S^{*})\right]$.
If $S^{*}$ is the seed set that maximizes ${\Exp}\left[C_g(S)\right]$, then,
$\optA = \{(u,i')\mid u \in S^{*}\}$ is the assignment that maximizes
${\Exp}\left[\sum_{v \in V} f_v(I_v^g(A)) \right]$.
Thus, solving the influence maximization problem and obtaining $S^{*}$,
yields the optimal solution $\optA$ for Problem~\ref{pr:diversityMax}.
\end{proof}

Given the monotonicity and submodularity of the objective function,
a standard greedy algorithm can be used to solve Problem~\ref{pr:diversityMax}.
The pseudocode is given in Algorithm~\ref{alg:basicGreedy}.
Let $\xgaindivscore{(u,i)}{A} = \xdivscore{A \cup \{(u,i)\}} -  \xdivscore{A}$
denote the marginal increase in the expected diversity exposure score of an assignment $A$
if $(u,i)$ is added to $A$.
Let $\greedyA$ denote the greedy solution.
At each iteration, the greedy algorithm chooses the feasible pair $(u^*, i^*)$
that yields the maximum gain in the expected diversity exposure score among all the feasible
pairs.\!\footnote{We say that a pair $(u,i)$ is feasible if it can be added to the current assignment $\greedyA$ without breaking the attention bound constraint $k_u$.}
The algorithm terminates 
when $\lvert \greedyA \rvert = k$.

\begin{algorithm}[t]
\caption{Greedy Algorithm}
\label{alg:basicGreedy}
\Indm
{\small
\SetKwInOut{Input}{Input}
\SetKwInOut{Output}{Output}
\SetKwComment{tcp}{//}{}
\Input{$\multiG = (V, \multiE, \multiP)$;
size constraint $k$; attention bound constraint $k_u$ for all $u \in V$;
leanings $\ell(i)$ for all $i \in H$, and $\ell(u)$ for all $u \in V$}
\Output{Greedy solution $\greedyA$}
}
\Indp
{\small
$\greedyA \leftarrow \emptyset$ \\
}
\BlankLine
\While{$\lvert \greedyA \rvert \le k$} {
		$(u^*, i^*) \leftarrow \underset{(u,i)} {\argmax\,} \xgaindivscore{(u,i)}{\greedyA}$, subject to: \ $\lvert \{i: (u,i) \in \greedyA \} \rvert \le k_u$ \label{line:sharpHard} \\
		$\greedyA \leftarrow \greedyA \cup \{(u^*,i^*)\}$ \\
	}
	{\bf return} $\greedyA$
\end{algorithm}

Before we analyze the approximation guarantee of the greedy algorithm,
we remind the reader of the following notions.

\begin{definition}[Matroid]
\label{def:Matroid}
A set system $(\calE,\calF)$,
defined over a finite ground set $\calE$ and a family $\calF$ of subsets of $\calE$,
is a matroid $\calM = (\calE, \calF)$ if
\squishlist
\item[($i$)] $\calF$ is non-empty;
\item[($ii$)] $\calF$ is downward closed,
i.e., $X \in \calF$ and $Y \subseteq X$ implies $Y \in \calF$; and
\item[($iii$)] $\calF$ satisfies the augmentation property,
i.e., for all $X,Y \in \calF$ with $|Y| > |X|$,
there exists an element $e \in Y \setminus X$
such that $X \cup \{e\} \in \calF$.
\squishend
\end{definition}

\begin{definition}[Uniform Matroid]
\label{def:uniMatroid}
A matroid $\calM = (\calE, \calF)$ is a uniform matroid
if $\calF = \{X \subseteq \calE : \lvert X \rvert \le k \}$.
\end{definition}

\begin{definition}[Partition Matroid]\label{def:PartitionMatroid}
Let $\calE_1, \cdots, \calE_Z$ be a partition of the ground set $\calE$ into $Z$ non-empty disjoint subsets.
Let $d_z$ be an integer with $0 \le d_z \le \lvert\calE_z\rvert$, for each $z=1,\ldots,Z$.
A matroid $\calM = (\calE, \calF)$ is a partition matroid if
$\calF = \{X \subseteq \calE : \lvert X \cap \calE_z \rvert \le d_z, \mbox{ for all } z = 1,\cdots,Z \}$.
In other words, a partition matroid contains exactly the sets $X \subseteq \calE$
that share at most $d_z$ elements with each subset $\calE_z$.
\end{definition}

\begin{lemma}\label{lem:constraintMatroid}
Given the ground set $\calE = V \times H$ of user$\,\times\,$item assignments,
an integer $k$, and integers $k_u$, for all $u\in  V$,
let $\calF \subseteq 2^{\calE}$ denote the set of feasible solutions to Problem~\ref{pr:diversityMax}.
Then, $\calM = (\calE, \calF)$ is a matroid defined on $\calE$.
\end{lemma}

\begin{proof}
To prove this result, we will show that
(i) the constraint $\lvert A \rvert \le k$ corresponds to a uniform matroid defined on~$\calE$,
which we denote by $\calM_k$;
(ii) the constraints $\lvert\{i : (u,i) \in A \}\rvert \le k_u$,
for all $u \in V$, correspond to a partition matroid defined on~$\calE$, which we denote by $\calM_p$;
(iii) the intersection $\calM= \calM_k \cap \calM_p$ is also a matroid defined on $\calE$.

Let $\calF_k \subseteq 2^{\calE}$ denote the set of assignments of size at most $k$,
i.e., $\calF_k = \{A \subseteq \calE: \lvert A \rvert \le k \}$.
It is easy to see that $\calM_k = (\calE, \calF_k)$ is a uniform matroid.

Let $\calF_p \subseteq 2^{\calE}$ denote the set of assignments
that do not violate any user attention bound constraint,
i.e., for all $A \in \calF_p$,
we have $\lvert \{i: (u,i) \in A \}\rvert \le k_u$, for all $u \in V$.
Define $\calE_u = \{(u,i): i \in H\}$, for all $u \in V$.
The sets $\calE_u$, with $u \in  V$, form a partition of $\calE$ into $n$ disjoint sets.
Notice that an assignment $\mathcal{A} \subseteq \calE$ can belong in $\calF_p$ if and only if
\begin{align*}
\lvert \mathcal{A} \cap \calE_u \rvert \le k_u, \mbox{ for all } u\in V.
\end{align*}
Hence, $\calM_p = (\calE, \calF_p)$ is a partition matroid.

Notice that the set $\calF$ of feasible solutions to Problem~\ref{pr:diversityMax}
is given by $\calF = \calF_k \cap \calF_p$.
Hence, the set system $(\calE, \calF)$ corresponds to the intersection of matroids
$\calM_k$ and $\calM_p$ that are both defined on $\calE$.
Note that the intersection of two matroids is not necessarily a matroid in general. However, in this case we have the intersection of a matroid with a uniform matroid, which is known to always result in a matroid; this operation is known as the \textit{truncation} of a matroid~\cite{Cygan2015parameterized}.
\end{proof}

\begin{theorem}
\label{theo:approximation}
Algorithm~\ref{alg:basicGreedy} achieves an approximation guarantee of $1/2$.
\end{theorem}

\begin{proof}
As we have shown in Lemma~\ref{lem:constraintMatroid},
the constraints of Problem~\ref{pr:diversityMax} correspond to a matroid defined on the ground set~$\calE$.
Moreover, we have shown in Lemma~\ref{lemma:submodularity}
that the objective function of Problem~\ref{pr:diversityMax} is monotone and submodular.
Thus, Problem~\ref{pr:diversityMax} corresponds to monotone submodular function maximization
subject to a matroid constraint.

Therefore, the approximation guarantee of Algorithm~\ref{alg:basicGreedy}
thus follows from the result of Fisher et al.~\cite{fisher1978analysis} for submodular function maximization subject to a matroid constraint. 
\end{proof}

\section{Scalable approximation algorithms}
\label{sec:algorithms}
The efficient implementation of the greedy algorithm (Algorithm~\ref{alg:basicGreedy}) is a challenge as the operation on line~\ref{line:sharpHard} translates to a large number of expected spread computations: in each iteration, the greedy algorithm requires to compute the expected marginal gain
$\xgaindivscore{(u,i)}{\greedyA}$ for every feasible pair $(u,i)$,
which in turn requires to identify the set $I_v^g(A \cup \{(u,i)\})$
of items that every $v$ is exposed to in each $g \sqsubseteq \multiG$, which is akin to computing the expected influence spread when $h=1$.

Computing the expected influence spread of a given set of nodes
under the independent-cascade model is \SPhard~\cite{ChenWW10}.
A common practice is to estimate the expected spread using Monte Carlo (MC) simulations~\cite{kempe03}.
However, accurate estimation requires a large number of MC simulations.

Hence, 
considerable effort has been devoted in the literature
to developing scalable approximation algorithms.
Recently, Borgs et al.~\cite{borgs14} introduced the idea of sampling \emph{reverse-reachable} sets (\rr-sets),
and proposed a quasi-linear time randomized algorithm.
Tang et al.\ improved it to a near-linear time randomized algorithm,
called {\em Two-phase Influence Maximization} (\tim)~\cite{tang14},
and subsequently tightened the lower bound on the number of random \rr-sets
required to estimate influence with high probability~\cite{tang2015influence}.

Random \rr-sets are critical for efficient estimation of the expected influence spread.
However, they are designed for the standard influence-maximization problem,
which is a special case of Problem~\ref{pr:diversityMax}.
We introduce a non-trivial generalization of reverse-reachable sets,
which we name \emph{reverse co-exposure} sets ({\rc}-sets), and
devise estimators for accurate estimation of the expected diversity exposure score $\xdivscore{\cdot}$.

\subsection{Reverse co-exposure sets}

Recall that we can interpret $\multiG$ as a probability distribution
over all subgraphs of $(V, \multiE)$,
where each edge $(u,v)_i \in \multiE$ is realized with probability $p_{uv}^i$.
Let $g \sim \multiG$ be a graph drawn from the random graph distribution $\multiG$.
Notice that, over the randomness in $g$,
the set $I_v^g(A)$ can be regarded as a Multinoulli random variable with $2^h$ outcomes,
where each outcome corresponds to one of the subsets of $H$.
Now, let $\multiR_{v,g} \subseteq \calE$ denote the set of pairs in $g$
that can color-reach~$v$, i.e.,
$\multiR_{v,g} = \{(u,i) \in \calE : \p_g^i(u,v) = 1\}$.
Also let
\[
I(A \cap \multiR_{v,g}) = \{i \in H : (u,i) \in  A \cap \multiR_{v,g}\}.
\]

The following lemma establishes the activation equivalence property
that forms the foundations of random \emph{reverse co-exposure} sets ({\rc}-sets).

\begin{lemma}\label{lemma:actEquivalence}
Let $I$ be a subset of $H$.
For any assignment $A$ and for all $v \in V$, we have
\[
\prob_{g \sim \multiG}\left(I_v^g(A) = I \right) \;=\; \prob_{g \sim \multiG} (I(A \cap \multiR_{v,g}) = I).
\]
\end{lemma}

\begin{proof}
Notice that in any possible world $g$, we have:
\begin{eqnarray*}
I_v^g(A)
&=& \{i \in H: \exists\, (u,i) \in A \text{ such that } \p_g^i(u,v) = 1\} \\
&=& \{i \in H : (u,i) \in  A \cap \multiR_{v,g}\} \\
&=& I(A \cap \multiR_{v,g}).
\end{eqnarray*}
Hence we have
\begin{eqnarray*}
\prob_{g \sim \multiG}\left(I_v^g(A) = I \right)
&=& \sum_{g \sqsubseteq \probG} Pr[g] \, \mathbbm{1}_{[I_g(A) = I]} \\
&=& \sum_{g \sqsubseteq \probG} Pr[g] \, \mathbbm{1}_{[I(A \cap \multiR_{v,g}) = I]} \\
&=& \prob_{g \sim \multiG} (I(A \cap \multiR_{v,g}) = I).
\end{eqnarray*}
\end{proof}
Next we formally define the concept of random {\rc}-sets.

\spara{Random~{\rc}-sets.}
Given a probabilistic multi-graph $\multiG = (V,\multiE, \multiP)$ and a set $H$ of items,
a random {\rc}-set $\multiR_{v,g}$ is generated as follows.
First, we remove each edge $(u,v)_i$ from $\multiG$ with probability $1-p^i_{uv}$, generating thus a possible world~$g$.
Next, we pick a \emph{target} node $v$ uniformly at random from $V$.
Then, $\multiR_{v,g}$ consists of the pairs that can \emph{color-reach} $v$,
i.e., all pairs $(u,i)$  
for which $\p_g^i(u,v) = 1$.

Sampling a random {\rc}-set $\multiR_{v,g}$ can be implemented efficiently
by first choosing a target node $v \in V$ uniformly at random and
then performing a breadth-first search (BFS) from $v$ in $\multiG$.
Notice that a random {\rc}-set $R_{v,g}$ is subject to two levels of randomness:
(i) randomness over $g \sim \multiG$, and
(ii) randomness over the selection of target node $v \sim V$.

\begin{lemma}
\label{lemma:RCEstimation}
For any random {\rc}-set $R_{v,g}$, 
let the random variable
$w(A \cap \multiR_{v,g}) = f_v(I(A \cap \multiR_{v,g}))$
represent the diversity exposure weight of $A$ on $\multiR_{v,g}$.
Then,
$\xdivscore{A} = n\, \underset{v,g} {\Exp} \left[w(A \cap \multiR_{v,g}) \right],$
where the expectation is taken over the randomness in $v \sim V$ and $g \sim \multiG$.
\end{lemma}

\begin{proof}
First, notice that over the randomness in $g$,
$f_v(I_v^g(A))$ is a function of a random variable $I_v^g(A)$, hence,
by the {\sc\large lotus} theorem~\cite{ross1970applied},
which defines expectation for functions of random variables,
its expectation can be computed as
\begin{align}
\label{eq:lotus}
\underset{g}{\Exp} \left[f_v(I_v^g(A))\right] =
\sum_{I \in 2^{H}} \underset{g}{\prob}\left(I_v^g(A) = I \right) \, f_v(I).
\end{align}
Then, by Equation~(\ref{eq:lotus}) and the activation equivalence property
shown in Lemma~\ref{lemma:actEquivalence}, we have
\begin{eqnarray*}
\xdivscore{A}
&=& \underset{g}{\Exp}\left[\sum_{v \in V} f_v(I_v^g(A))\right] \\
&=& \sum_{v \in V} \underset{g}{\Exp} \left[f_v(I_v^g(A))\right]\\
&=& \sum_{v \in V} \sum_{I \in 2^H} \prob_g\left(I_v^g(A) = I \right) \, f_v(I) \\
&=& n \, \sum_{I \in 2^{H}} \prob_{v,g}(I(A \cap \multiR_{v,g}) = I)  \, f_v(I) \\
&=& n \, \underset{v,g} {\Exp} \left[f_v(I(A \cap \multiR_{v,g})) \right].
\end{eqnarray*}
\end{proof}
Lemma~\ref{lemma:RCEstimation} shows that we can estimate $\xdivscore{A}$
by estimating $n \, {\Exp} \left[f_v(I(A \cap \multiR_{v,g})) \right]$ on a set of random {\rc}-sets.
This suggests that if we have a sample $\sampleR$ of random {\rc}-sets from which we can obtain,
with high probability,
accurate estimations of $\xdivscore{A}$ for every assignment $A$
such that $\lvert A \rvert \le k$, then,
we can accurately solve Problem~\ref{pr:diversityMax}
on the sample $\sampleR$ with high probability, as we show next.

Given a sample $\sampleR$ of random {\rc}-sets, let
\[
\mathcal{W}_{\sampleR}(A) =
\frac{\sum_{\multiR_{v,g} \in \sampleR} w(A \cap \multiR_{v,g})}{\lvert \sampleR \rvert},
\]
denote the diversity exposure weight of $A$ on the sample.
Notice that, as a direct consequence of Lemma~\ref{lemma:RCEstimation},
the quantity $n \, \mathcal{W}_{\sampleR}(A)$ is an unbiased estimator of $\xdivscore{A}$.

Moreover, let
\[
\mathcal{W}_{\sampleR}((u,i) \mid A) = \mathcal{W}_{\sampleR}(A \cup \{(u,i)\}) - \mathcal{W}_{\sampleR}(A),
\]
denote the marginal increase in the diversity exposure weight of $A$ if the pair $(u,i)$ is added to $A$.
\begin{algorithm}[t]
\caption{\fastAlgo$(\multiG, k, l, \epsilon, \ell)$}
\label{alg:tdem}
\Indm
{\small
\SetKwInOut{Input}{Input}
\SetKwInOut{Output}{Output}
\SetKwComment{tcp}{//}{}
}
\Indp
\BlankLine
$\sampleR \gets \mathrm{Sampling}(\multiG, k, \epsilon, \ell)$ \\
$\approxgrA \gets \rcGreedy(\sampleR, k, l)$ \\
{\bf return} $\approxgrA$
\end{algorithm}

\begin{algorithm}[t]
\caption{$\rcGreedy(\sampleR, k, l)$}
\label{alg:rcGreedy}
\Indm
{\small
\SetKwInOut{Input}{Input}
\SetKwInOut{Output}{Output}
\SetKwComment{tcp}{//}{}
}
\Indp
\BlankLine
$\approxgrA \gets \emptyset$ \\
\While{$\lvert \approxgrA \rvert \le k$} {
		$(u^*, i^*) \leftarrow {\argmax_{(u,i)}} \mathcal{W}_{\sampleR}((u,i) \mid \approxgrA)$, subject to: \ $\lvert \{i: (u,i) \in \approxgrA \} \rvert \le k_u$ \\
		$\approxgrA \leftarrow \approxgrA \cup \{(u^*,i^*)\}$ \\
}
{\bf return} $\approxgrA$
\end{algorithm}

\subsection{\fastAlgoLong}
We now present our \fastAlgoLong algorithm (\fastAlgo), 
which provides an approximate solution to Problem~\ref{pr:diversityMax}.
The pseudocode is shown in Algorithm~\ref{alg:tdem}.
As it names suggests, \fastAlgo operates in two phases:
a \emph{sampling} phase and a \emph{greedy pair-selection} phase.
In the sampling phase, a sample $\sampleR$ of random {\rc}-sets is generated (details later).
This sample is provided as input to \rcGreedy (Algorithm~\ref{alg:rcGreedy}),
which greedily selects feasible pairs $(u, i)$ into $\approxgrA$.
The algorithm terminates when $\lvert \approxgrA \rvert = k$ and
it returns $\approxgrA$ as a solution to Problem~\ref{pr:diversityMax}.

\begin{theorem}
\label{theo:rcGreedy}
Assume that the algorithm \rcGreedy receives as input
a sample $\sampleR$ of random {\rc}-sets such that
for any assignment $A$ of size at most $k$ it holds that
\begin{align}\label{eq:theoAccuracy}
\left\lvert n \, \mathcal{W}_{\sampleR}(A) - \xdivscore{A} \right\rvert < \dfrac{\epsilon}{2} \, \optvalue,
\end{align}
with probability at least $1 -n^{-\ell} / \binom{n h}{k}$. Then, \rcGreedy returns a $(\frac{1}{2} - \epsilon)$-approximate solution to Problem~\ref{pr:diversityMax} with probability at least $1 - n^{-\ell}$.  The running time of \rcGreedy is $\bigO(\sum_{\multiR \in \sampleR} \lvert \multiR \rvert)$, that is, linear in the total size of the {\rc}-sets in the sample.
\end{theorem}
\begin{proof}
First, notice that, $\mathcal{W}_{\sampleR}(\cdot)$ is a linear combination of submodular $f_v(\cdot)$'s, hence is also submodular. Moreover, the activation equivalence property depicted in 
in Lemma~\ref{lemma:actEquivalence} shows that we can approximately solve Problem~\ref{pr:diversityMax} by finding the assignment that maximizes $\mathcal{W}_{\sampleR}(\cdot)$ on a sample $\sampleR$ of {\rc}-sets. Now, let $$\approxgrA^* = \underset{A \in \mathcal{F}}{\argmax}\, \mathcal{W}_{\sampleR}(A).$$ Then by submodularity and monotonicity of $\mathcal{W}_{\sampleR}(\cdot)$, we have
\begin{align*}
\mathcal{W}_{\sampleR}(\approxgrA) \ge \frac{1}{2} \mathcal{W}_{\sampleR}(\approxgrA^*).
\end{align*}
Given that  $\approxgrA^*$ is the optimal solution on the sample, we also have $\mathcal{W}_{\sampleR}(\approxgrA^*) \ge \mathcal{W}_{\sampleR}(A^*)$ where $A^*$ is the optimal solution of Problem~\ref{pr:diversityMax}. Reminding that Eq.\ref{eq:theoAccuracy} holds for any assignment of size at most $k$ w.p. at least $1 -n^{-\ell} / \binom{n h}{k}$, by a union bound, w.p. at least $1 - 1 / n^{\ell}$, we have: 
\begin{align*}
\xdivscore{\approxgrA} &\ge n \, \mathcal{W}_{\sampleR}(\approxgrA) - \dfrac{\epsilon}{2} \, \optvalue, \\
&\ge \frac{1}{2} \mathcal{W}_{\sampleR}(\approxgrA^*)  - \dfrac{\epsilon}{2} \, \optvalue, \\
&\ge \frac{1}{2} \mathcal{W}_{\sampleR}(A^*) - \dfrac{\epsilon}{2} \, \optvalue, \\
&\ge \frac{1}{2}( \xdivscore{A^*} - \dfrac{\epsilon}{2} \, \optvalue) - \dfrac{\epsilon}{2} \, \optvalue, \\
&= \frac{1}{2}(OPT - \dfrac{\epsilon}{2} \, \optvalue) - \dfrac{\epsilon}{2} \, \optvalue, \\
&\ge \frac{1}{2}OPT - \epsilon \, \optvalue. 
\end{align*}

Therefore the result follows.

Now we analyze the running time of \rcGreedy. First, we remind that the running of the greedy algorithm on {\rr} sets, for approximately solving the influence maximization problem, follows from the running time of the maximum cover problem~\cite{tang14}. For the analysis of \rcGreedy, we use a similar reasoning and exploit a connection to the weighted version of the maximum coverage problem. However, we note that our problem does not correspond to the weighted maximum coverage problem since (i) we are interested in the weights of {\rc}-sets even in the case when they have been already covered by a pair $(u,i)$,\!\footnote{We say that a pair $(u,i)$ covers an {\rc}-set $\multiR$ if $(u,i) \in \multiR$} (ii) the weights of the ground set elements (which correspond to {\rc}-sets) dynamically change based on the pairs that already covered them. However, these differences do not affect the running time analysis much. The constant time operation to check whether an {\rc}-set is covered by a pair $(u,i)$ is replaced by the operation of finding the next smaller and next larger labels compared to $l(i)$ from the labels of the items that have previously covered this {\rc}-set. Using binary search, this can be done in logarithmic time. 

Since this operation is independent of the seed node $u$, the number of ``covered" checks performed on each {\rc}-set is upper-bounded by the size of the {\rc}-set, times a logarithmic factor as explained above. Hence, the total running time complexity of \rcGreedy is $\bigO(\sum_{\multiR \in \sampleR} \lvert \multiR \rvert \log (\lvert \multiR \rvert))$. 
\end{proof}

Let $\theta^*$ be the minimum sample size such that Equation~(\ref{eq:theoAccuracy})
holds for all assignments of size at most $k$.
Notice that since the desired estimation accuracy is a function of $\optvalue$,
the value of $\theta^*$ also depends on $\optvalue$, which is unknown and in fact \NPhard to compute.
To circumvent the problem
we follow a similar approach to \tim \cite{tang14} and \imm \cite{tang2015influence}:
we estimate a lower bound
on the value of the optimal solution,
and use it for the determination of the sample size. 
We also generalize the \emph{statistical test} employed by \imm\cite{tang2015influence} for estimating a lower bound when working with random {\rc}-sets.
Note that 
the results from
influence maximization do not carry over to our case, therefore our extension of the technique is non-trivial.

\subsection{Determining the sample size}
Let $\multiR_1, \multiR_2, \ldots, \multiR_{\theta}$
be the sequence of random \rc-sets generated in the sampling phase of \fastAlgo.
For a given assignment $A$, let $w_j$ denote its weight on the \rc-set $\multiR_j$.
Notice that the choices of $v$ and $g$ during the creation of $\multiR_j$
are independent of $\multiR_1, \ldots, \multiR_{j-1}$.
However, as we will see soon, 
the sampling phase of \fastAlgo employs an adaptive procedure,
in which the decision to generate $\multiR_j$ depends on the outcomes of
$\multiR_1, \ldots, \multiR_{j-1}$.
This creates dependencies between the \rc-sets in the sample $\multiR$.
Thus, 
we can only use concentration inequalities that allow dependencies in the sample.
We first introduce the notions that are crucial in our analysis.

\begin{definition}[Martingale]\label{def:Martingale}
A sequence $X_1, X_2, \ldots$ of random variables is a martingale
if and only if $\Exp[\lvert X_j \rvert]  < + \infty$ and
$\Exp[X_j \mid X_1, \ldots, X_{j-1}] = X_{j-1}$ for any $j$.
\end{definition}

We now establish the connections to martingales.
Let $w = \xdivscore{A}/n$.
By Lemma~\ref{lemma:RCEstimation} we have $\Exp[w_j] = w$, for all $j \in [1,\theta]$.
Noting that the choice of $v$ and $g$ during the creation of $\multiR_j$
is independent of $\multiR_1, \ldots, \multiR_{j-1}$, we have 
\[
\Exp[w_j \mid w_1, \ldots, w_{j-1}] = \Exp[w_j] = w.
\]

Let $M_j= \sum_{z=1}^{j}(w_z - w)$, so $\Exp[M_j] = 0$, and
\begin{align*}
\Exp[M_j \!\mid\! M_1, \ldots, M_{j-1}] &= \Exp[M_{j-1} + w_j - w \!\mid\! M_1, \ldots, M_{j-1}] \\
&= M_{j-1} - w + \Exp[w_j \!\mid\! M_1, \ldots, M_{j-1}] \\
&= M_{j-1} - w + \Exp[w_j] \\
&= M_{j-1},
\end{align*}
therefore, 
the sequence $M_1, \ldots, M_{\theta}$ is a martingale.

The following lemma from  Chung and Lu~\cite{chung2006concentration}
shows a con\-cen\-tra\-tion result for martingales,
analogous to Chernoff bounds for 
independent random variables.
\begin{lemma}\label{lemmaMartIneqLiterature}[Theorem $6.1$, \cite{chung2006concentration}]
Let $X_1, X_2, \ldots $ be a martingale,
such that $X_1 \le a$, $\Var[X_1] \le b_1$, $\lvert X_z - X_{z-1} \rvert \le a$ for $z \in [2,j]$, and
\[
\Var[X_z \mid X_1, \ldots, X_{z-1}] \le b_j, \text{ for } z \in [2,j],
\]
where $\Var[\cdot]$ denotes the variance. Then, for any $\gamma > 0$
\begin{align*}
\prob\left(X_j - \Exp[X_j] \ge \gamma \right) \le \exp\left(-\dfrac{\gamma^2}{2(\sum_{z = 1}^{j} b_z + a \gamma /3)} \right)
\end{align*}
\end{lemma}

We now discuss how to use this concentration result for the martingale $M_1, \ldots, M_{\theta}$.
Notice that since $w_j \in [0,1]$ for all $j \in [1, \theta]$,
we have $\lvert M_1 \rvert= \lvert w_1 - w \rvert \le 1$ and
$\lvert M_j - M_{j-1} \rvert \le 1$ for any $j \in [2,\theta]$.
We also have $\Var[M_1] = \Var[w_1]$, and for any $j \in [2,\theta]$
\begin{align*}
\Var[M_j  & \mid M_1, \ldots, M_{j-1}] \\
&=  \Var[M_{j-1} + w_j - w \mid M_1, \ldots, M_{j-1}] \\
&=\Var[w_j \mid M_1, \ldots, M_{j-1}]  \\
&=\Var[w_j].
\end{align*}

Recall that $f_v(I_v^g(A))$ is a function of the
Multinoulli random variable $I_v^g(A)$, hence,
$w(A \cap \multiR_{v,g}) = f_v(I(A \cap \multiR_{v,g}))$.
Based on the {\sc\large lotus} theorem~\cite{ross1970applied} again,
we have
\begin{align*}
\Exp[f_v(I(A \cap \multiR_{v,g}))^2] = \sum_{I \in 2^{H}} \prob_{v,g}(I(A \cap \multiR_{v,g}) = I)  \, (f_v(I))^2.
\end{align*}

Hence, we can bound the variance as follows
\begin{align*}
\Var[f_v(I(A \cap \multiR_{v,g}))] 
& = \Exp[f_v(I(A \cap \multiR_{v,g}))^2] - w^2  \\
&\le \Exp[f_v(I(A \cap \multiR_{v,g}))^2] \\
&= \sum_{I \in 2^{H}} \prob_{v,g}(I(A \cap \multiR_{v,g}) = I) \, (f_v(I))^2 \\
&\le \sum_{I \in 2^{H}} \prob_{v,g}(I(A \cap \multiR_{v,g}) = I) \, f_v(I) \\
&= w,
\end{align*}
where the last inequality follows from the fact that $f_v(\cdot)$ is bounded by $1$.
Therefore, $\Var[w_j] \le w$ for all $j \in [1, \theta]$.
Then, by using Lemma~\ref{lemmaMartIneqLiterature}, for
$M_{\theta}=  \sum_{j=1}^{\theta}(w_j - w)$, with
$\Exp[M_{\theta}] = 0$, $a = 1$, $b_j=w$, for $j  = 2, \ldots ,\theta$, and
$\gamma = \delta \theta w$, we have the following corollary.
\begin{corollary}
\label{corr:ourIneq1}
For any $\delta > 0$,
\begin{align*}
\prob\left( \sum_{j=1}^{\theta} w_j - \theta w \ge \delta \theta w \right) \le \exp\left(-\dfrac{\delta^2}{\frac{2\delta}{3} + 2} \, \theta w \right).
\end{align*}
\end{corollary}
Moreover, for the martingale $-M_1, \ldots, -M_{\theta}$,
we similarly have $a = 1$ and $b_j = w$ for $j  = 1, \ldots ,\theta$.
Note also that $\Exp[-M_{\theta}] = 0$.
Hence, 
for $-M_{\theta}=  \sum_{j=1}^{\theta}(w - w_j)$ and $\gamma = \delta \theta w$ we can obtain:
\begin{corollary}\label{corr:ourIneq2}
For any $\delta > 0$,
\begin{align*}
\prob\left( \sum_{j=1}^{\theta} w_j - \theta w \le -\delta \theta w \right) \le \exp\left(-\dfrac{\delta^2}{\frac{2\delta}{3} + 2} \, \theta  w \right).
\end{align*}
\end{corollary}
We will use these corollaries frequently.
We are now ready to start our analysis.
We first provide a lower bound on the sample size, which depends on $\optvalue$.

\begin{lemma}\label{lemma:thetaLbOpt}
Let $\theta = \lvert \sampleR \rvert$ denote the size of the random \rc-sets
returned by the sampling phase of \fastAlgo.
Suppose that~$\theta$ satisfies
\begin{align}\label{eq:thetaLbOpt}
\theta \ge 4n(\epsilon + 6)\,\dfrac{\ln \binom{n h}{k} + \ell \ln n + \ln 2}{3 \epsilon^2 \, \optvalue}.
\end{align}
Then, for any assignment $A$ of size at most $k$,
the following holds with probability at least $1 - n^{-\ell}/\binom{n h}{k}$
\begin{align}\label{eq:accuracy}
\left\lvert n \, \mathcal{W}_{\sampleR}(A) - \xdivscore{A} \right\rvert < \dfrac{\epsilon}{2} \, \optvalue.
\end{align}
\end{lemma}

For better readability, we have included the proof of Lemmas \ref{lemma:thetaLbOpt},
\ref{lemma:optLB1}, and \ref{lemma:optLB2} in the Supplementary material. 

As stated in Theorem~\ref{theo:rcGreedy}
the greedy pair selection phase of \fastAlgo requires as input a sample $\sampleR$ of
random \rc-sets such that Equation~(\ref{eq:theoAccuracy}) holds for all assignments of size at most~$k$.
As shown in Lemma~\ref{lemma:thetaLbOpt},
this requirement translates to the lower bound
$\lvert \sampleR \rvert \ge \lambda / \optvalue$, where
\begin{align}\label{eq:lambda}
\lambda = 4n (\epsilon + 6)  \, \dfrac{\ln \binom{n h}{k} + \ell \ln n + \ln 2}{3 \epsilon^2}.
\end{align}
Given that $\optvalue$ is unknown and \NPhard to compute,
our objective is to identify a lower bound on $\optvalue$,
which is as tight as possible,
so as to reduce the computational cost of generating the sample $\multiR$.
To achieve this goal,
we extend the technique introduced by \imm and
we perform a statistical test $B(x)$,
such that if $\optvalue < x$ then $B(x) = \mathtt{false}$ with high probability.
Given that $\optvalue \in (0,n]$ and
using the value of the greedy solution as an indicator of the magnitude of $\optvalue$,
we can identify a lower bound on $\optvalue$ by running the test $B(x)$ on
$\bigO(\log_2 n)$ values of $x$, i.e., $x = n/2, n/4, \ldots, 2$.

We now give details of our sampling algorithm,
which first adaptively estimates a lower bound on the value of $\optvalue$
by employing the statistical test,
and then it keeps generating random {\rc}-sets into $\multiR$
until $\lvert \multiR \rvert \ge \lambda/ \LB$.

The sampling algorithm, pseudocode provided in Algorithm~\ref{alg:sampling}, 
first sets $\sampleR = \emptyset$ and
initializes $\LB$ to a na\"ive lower bound --- which we will explain soon.
Then, it enters a $\mathtt{for}$-loop with at most $\log_2 n$ iterations.
In the $i$-th iteration, the algorithm computes $x = n / 2^i$ and derives
\[
\theta_i  = \frac{(\frac{2\epsilon}{3} + 2) \, \left( \ln \binom{n h}{k} + \ell \ln n + \ln \log_2 n \right)}{\epsilon^2}  \, \frac{n}{x}.
\]
Then the Algorithm inserts more random {\rc}-sets into $\sampleR$
until $\lvert \sampleR \rvert \ge \theta_i$ and
invokes \rcGreedy (Algorithm~\ref{alg:rcGreedy}).
If $\sampleR$ satisfies the following \emph{stopping condition}
\begin{align}\label{eq:stoppingCond}
n \, \mathcal{W}_{\sampleR}(\approxgrA) \ge (1 + \epsilon) \, x,
\end{align}
the algorithm sets the lower bound
$\LB = \frac{n \, \mathcal{W}_{\sampleR}(\approxgrA)}{1 + \epsilon}$ and
terminates the $\mathtt{for}$-loop.
If this is the case, then algorithm generates more random {\rc}-sets into $\multiR$
until $\lvert \multiR \rvert \ge \frac{\lambda}{\LB}$ and returns $\sampleR$.
Otherwise, the algorithm proceeds with the $(i+1)$-th iteration.
If after $\bigO(\log_2 n)$ iterations the algorithm cannot set $\LB$,
then it uses the na\"ive lower bound and generates random {\rc}-sets into $\sampleR$
until $\lvert \sampleR \rvert \ge \lambda / \LBnaive$.
The na\"ive bound $\LBnaive$ corresponds to the value of the minimum possible solution on the input instance for any positive integer $k$,
hence, we set $\LBnaive = 1- \frac{1}{4}\underset{(v,i) \in \calE}{\max~} g_v(\{i\})$.\!\footnote{Notice that this is analogous to \imm's naive lower bound for the influence maximization problem that is equal to $1$.}

\begin{algorithm}[t]
\caption{$\mathrm{Sampling}(\multiG, k, \epsilon, \ell)$}
\label{alg:sampling}
\Indm
{\small
\SetKwInOut{Input}{Input}
\SetKwInOut{Output}{Output}
\SetKwComment{tcp}{//}{}
}
\Indp
\BlankLine
\Indp
{\small
$\multiR \gets \emptyset$ \;
$\LB \gets \LBnaive$ \;
}
\For{$i =1, \ldots, \log_2 n - 1$} {
$x \gets n / 2^{i}$ \;
$\theta_i  = \frac{(\frac{2\epsilon}{3} + 2) \, \left( \ln \binom{n h}{k} + \ell \ln n + \ln \log_2 n \right)}{\epsilon^2}  \, \frac{n}{x}$ \;
\While{$\lvert \sampleR \rvert \le \theta_i$}{
$\sampleR \gets \sampleR \cup \mathrm{GenerateRC\text{-}Set}$\;
}
$\approxgrA_i \gets \rcGreedy(\sampleR, k, l)$ \;
\If{$n \, \mathcal{W}_{\sampleR}(\approxgrA_i) \ge (1 + \epsilon) \, x,$}{
	$\LB \gets \frac{n \, \mathcal{W}_{\sampleR}(\approxgrA)}{1 + \epsilon}$ \;
	{\bf break}\;
}
$\theta \gets \lambda /\LB$\;
}
\While{$\lvert \sampleR \rvert \le \theta$}{
$\sampleR \gets \sampleR \cup \mathrm{GenerateRC\text{-}Set}$ \;
}
{\bf return} $\sampleR$
\end{algorithm}

The following theorem gives the correctness of Algorithm~\ref{alg:sampling}.
\begin{theorem}~\label{theorem:stopCondCorrectness}
With probability at least $1 - n^{-\ell}$, Algorithm~\ref{alg:sampling} returns a sample $\sampleR$ such that $\lvert \sampleR \rvert \ge \lambda / \optvalue$.
\end{theorem}

To prove Theorem~\ref{theorem:stopCondCorrectness}, we first establish the following two~lemmas, 
whose proof can be found in the Supplementary material.

\begin{lemma}\label{lemma:optLB1}
Assume that we invoke algorithm \rcGreedy on a sample $\sampleR$ of $\theta$ random {\rc}-sets such that
\begin{align*}
\theta \ge \frac{(\frac{2\epsilon}{3} + 2) \, \left( \ln \binom{n h}{k} + \ell \ln n + \ln \log_2 n \right)}{\epsilon^2}  \, \frac{n}{x}.
\end{align*}

Let $\approxgrA$ be the solution returned by the \rcGreedy.
If $n \, \mathcal{W}_{\sampleR}(\approxgrA) \ge (1 + \epsilon) \, x$,
then $\optvalue \ge x$ with probability at least $1 - \frac{n^{-\ell}}{\log_2 2n}$.
\end{lemma}

\begin{lemma}
\label{lemma:optLB2}
Assume $x$, $\epsilon$, $\sampleR$, and $\approxgrA$ are defined as in Lemma~\ref{lemma:optLB1}.
If $\optvalue \ge x$ then
$n \, \mathcal{W}_{\sampleR}(\approxgrA) \le (1 + \epsilon) \, \optvalue$
with probability at least $1 - \frac{n^{-\ell}}{\log_2 n}$.
\end{lemma}

We are now ready prove Theorem~\ref{theorem:stopCondCorrectness}.

\begin{proof}[Proof of Theorem~\ref{theorem:stopCondCorrectness}]
Let $i^* = \lceil \log_2 \frac{n}{\optvalue} \rceil$.
We will first show that the probability the stopping condition holds
while $\optvalue < x$ is at most $(i^* -1)/(n^{\ell} \log_2 n)$.
Recall that the value of $x$ is determined by $n/2^i$ at each iteration $i$.
Therefore, for any $i < i^*$, we have $x = n/2^i < \optvalue$. Hence, by Lemma~\ref{lemma:optLB1} and the union bound over $i^* - 1$ iterations, the probability that $\optvalue < x$ and  $n \, \mathcal{W}_{\sampleR}(\approxgrA)/(1 + \epsilon) \ge x$  is at most $(i^*-1)/(n^{\ell} \log_2 n)$. Moreover, it follows from Lemma~\ref{lemma:optLB2} that the probability that $\optvalue \ge x$ and $n \, \mathcal{W}_{\sampleR}(\approxgrA) > (1 + \epsilon) \, \optvalue$  is at most $1/(n^{\ell} \log_2 n)$. Hence, when the stopping condition holds, by union bound, the probability that $\optvalue \ge x$ and $n \, \mathcal{W}_{\sampleR}(\approxgrA) \le (1 + \epsilon) \, \optvalue$ is at least 
\[1 - \left(\dfrac{i^*-1}{n^{\ell} \log_2 n} +  \dfrac{1}{n^{\ell} \log_2 n} \right) \ge 1 - n^{-\ell}.\] 
Then by  Lemma~\ref{lemma:optLB2} and the union bound, it follows that with probability at least $1 - n^{-\ell}$, we have
\[
\optvalue \ge \frac{n \, \mathcal{W}_{\sampleR}(\approxgrA)}{1 + \epsilon} \ge x.
\]
Therefore, the algorithm sets $\LB \ge \optvalue$ with probability at least $1 - n^{-\ell}$ and returns a sample $\sampleR$ such that
\begin{align*}
\lvert \sampleR \rvert \ge \frac{\lambda}{\LB} \ge \frac{\lambda}{\optvalue}
\end{align*}
with probability at least $1 - n^{-\ell}$.
\end{proof}

\section{Experiments}
\label{sec:experiments}
In this section, we evaluate our proposed algorithm on a range of real-world datasets.

\begin{table*}[ptb]
\centering
\caption{Statistics of the datasets.}
\label{tab:stats}
\begin{tabular}{@{\hspace{1ex}}l@{}l@{\hspace{6ex}}r@{\hspace{3ex}}r@{\hspace{3ex}}r@{\hspace{4ex}}r@{\hspace{3ex}}r@{\hspace{4ex}}r@{\hspace{3ex}}r@{\hspace{3ex}}r@{\hspace{1ex}}}
\toprule
Dataset &  & $\nbN{}$ & $\nbE{}$ & $\density{}$ & \multicolumn{1}{c}{$\ell$} & \multicolumn{1}{c}{$\ell^2$} & \multicolumn{3}{c}{$p^i_{uv}$} \\
 \cline{8-10} \\ [-.8em]
 &  & & & & $\avg$ & $\avg$ & $\min$ & $\avg$ & $\max$ \\
\midrule
\dtstDBLPbs{} &  & $167$ & $634$ & $3.80$ & $-0.60$ & $0.50$ & $0.034$ & $0.116$ & $0.249$ \\
\dtstDBLPcp{} &  & $144$ & $800$ & $5.56$ & $-0.26$ & $0.28$ & $0.034$ & $0.117$ & $0.247$ \\
\dtstDBLPpy{} &  & $342$ & $1964$ & $5.74$ & $-0.52$ & $0.42$ & $0.034$ & $0.118$ & $0.249$ \\
\dtstTWEPc{} &  & $140$ & $1372$ & $9.80$ & $-0.03$ & $0.34$ & $0.034$ & $0.112$ & $0.249$ \\
\dtstTWEPb{} &  & $338$ & $8436$ & $24.96$ & $-0.07$ & $0.43$ & $0.034$ & $0.107$ & $0.249$ \\
\dtstTWEPa{} &  & $577$ & $24427$ & $42.33$ & $0.12$ & $0.36$ & $0.034$ & $0.113$ & $0.250$ \\
\dtstTWTESFive{} &  & $2719$ & $7714$ & $2.84$ & $0.24$ & $0.52$ & $0.034$ & $0.114$ & $0.249$ \\
\dtstTWTESTwo{} &  & $4379$ & $27765$ & $6.34$ & $0.26$ & $0.52$ & $0.034$ & $0.113$ & $0.249$ \\
\dtstTWTEMFive{} &  & $5183$ & $50165$ & $9.68$ & $0.26$ & $0.51$ & $0.034$ & $0.113$ & $0.250$ \\
\dtstTWTS{} &  & $5454$ & $835725$ & $153.23$ & $0.27$ & $0.52$ & $0.034$ & $0.116$ & $0.250$ \\
\dtstBrexit{} &  & $22745$ & $48830$ & $2.15$ & $0.65$ & $0.72$ & $0.010$ & $0.014$ & $0.110$ \\
\dtstPhone{} &  & $36742$ & $49248$ & $1.34$ & $0.87$ & $0.90$ & $0.010$ & $0.053$ & $1.000$ \\
\dtstUselect{} &  & $23816$ & $844700$ & $35.47$ & $0.46$ & $0.75$ & $0.010$ & $0.013$ & $0.043$ \\
\dtstAbort{} &  & $279505$ & $670501$ & $2.40$ & $0.02$ & $0.80$ & $0.010$ & $0.011$ & $0.110$ \\
\dtstFrack{} &  & $374403$ & $1366909$ & $3.65$ & $0.55$ & $0.61$ & $0.010$ & $0.011$ & $0.110$ \\
\dtstObamacare{} &  & $334617$ & $1511670$ & $4.52$ & $0.12$ & $0.61$ & $0.010$ & $0.012$ & $0.110$ \\
\dtstTWTXL{} &  & $481523$ & $52378856$ & $108.78$ & $0.07$ & $0.39$ & $4.2\text{e-}5$ & $0.028$ & $1.000$ \\
\bottomrule
\end{tabular}
\end{table*}

\subsection{Datasets}
\label{exp:datasets}
In our experiments, we use five collections of networks, one based on data from the DBLP bibliographic database, and the other four datasets collected from Twitter. 

The first collection consists of the one-hop egonets of three well-known researchers: B.~Schneiderman (\dtstDBLPbs{}), P.~Yu (\dtstDBLPpy{}) and C.~Papadimitriou (\dtstDBLPcp{}). 
Node leanings are derived from publication-venue information using the method proposed by Galbrun et al.~\cite{galbrun17finding}.

\dtstTWTS{} is the Twitter follower network obtained by Lahoti et al.~\cite{Lahoti2018joint} and \dtstTWTXL{} is a larger variant of this same network. 
For node leanings we use rescaled ideology scores from Barber\'a et al.~\cite{barbera2015tweeting}. 
From the same harvest of tweets as \dtstTWTS{}, 
we construct two additional collections of networks.
The first collection contains the networks \dtstTWEPc{}, \dtstTWEPb{}, and \dtstTWEPa{}.
Each of these networks is obtained by selecting a pair of users who have opposite leanings but share neighbors and extracting the neighborhood.
The second collection contains the networks \dtstTWTESFive{}, \dtstTWTESTwo{}, and \dtstTWTEMFive{}.
Instead of follower-followee relationships, these networks capture actual exchanges of tweets between users, 
with increasing requirements on the strength of the exchanges.

The last collection consists of the six networks from the study of Garimella et al.~\cite{garimella17balancing}:
\dtstAbort{}, \dtstBrexit{}, \dtstFrack{}, \dtstPhone{}, \dtstObamacare{} and \dtstUselect{},
Each network represents a Twitter follower network focused around topics with two opposing sides. 
We obtain node leanings from estimated probabilities of users to retweet content from either of the opposing sides.

Note that solving our problem for $h$ items on a network with $m$ edges requires maintaining in memory a multigraph of $h \times m$ edges, which is analogous to the requirement for solving the standard influence maximization problem on a graph with $h \times m$ edges. Hence, our largest configuration, \dtstTWTXL{} with $25$ items, effectively yields a graph with $52.5\,\text{M} \times  25 \approx 1.3\,\text{B}$ edges and is comparable to \imm's largest dataset of $1.5\,\text{B}$ edges\cite{tang2015influence}. 

For the largest configuration \dtstTWTXL{}, following \imm\cite{tang2015influence}, we use the Weighted-Cascade model~\cite{kempe03} that assigns $p_{u,v}^i = 1/{|N^{in}(v)|}$ for each item and edge to retain comparability. For the rest of the datasets, the propagation probabilities of items along the edges of the network depend on the leaning of the item being propagated and the leanings of the emitting and receiving users. Intuitively, the further away from the leaning of the users, the less likely an item is to be propagated.

We consider an exponential function to model how the propagation probabilities drop as the leaning of the item lies further away from that of the communicating nodes. 
More specifically, we use an exponential function with parameters $\beta$ and $\gamma$:
\[\Phi_{\beta, \gamma}(u, v, i) = \beta \, \exp(-\gamma \, \max(\abs{\ell(u) - \ell(i)}, \abs{\ell(v) - \ell(i)})/2)\;.\]

We set $\beta = 0.25$ for all collections except \dtstClNIPS for which we use the edge probabilities present in the network as values for $\beta$ and add a $0.01$ offset to all resulting values, in order to obtain reasonable propagation probabilities. We experiment with probabilities obtained with the exponential function, letting $\gamma=2$. We compare the propagation probabilities resulting from this function to an exponential function with $\gamma=4$ as well as a linear function. A heatmap of the resulting propagation probabilities can be found in the Supplementary material.

We use $25$ items with leanings evenly spread over the interval $[-1,1]$ as our pool of items in all the setups.
For the smaller datasets, we will look for assignments of size $k=5$ with an attention bound $k_u = 1$, while for larger datasets we use $k=50$ and $k_u = 5$. We set $\epsilon = 0.2$ and $\ell = 1$ in all the experiments following~\cite{tang2015influence}.

Table~\ref{tab:stats} shows the basic statistics of the datasets used in our experiments.
For each dataset, we indicate the number of nodes ($\nbN{}$), the number of edges ($\nbE{}$), the density of the graph ($\density{}=\nbE{}/\nbN{}$), the average node leaning ($\ell$), the squared node leaning ($\ell^2$), as well as the minimum, average and maximum propagation probabilities, over all edges and items in the network ($p^i_{uv}$).

Figures~\ref{tab:app_histDBLP}--\ref{tab:app_histTWT} show histograms of node leanings and leaning differences across the edges of each network from the different collections. 

\begin{figure*}[pbt]
	\centering
	\begin{tikzpicture}[node distance=1.1em]
	\pgfmathsetmacro{\x}{4.7}%
	\node (bNIPS) at ({\x},0.000) {\includegraphics[width=.45\textwidth]{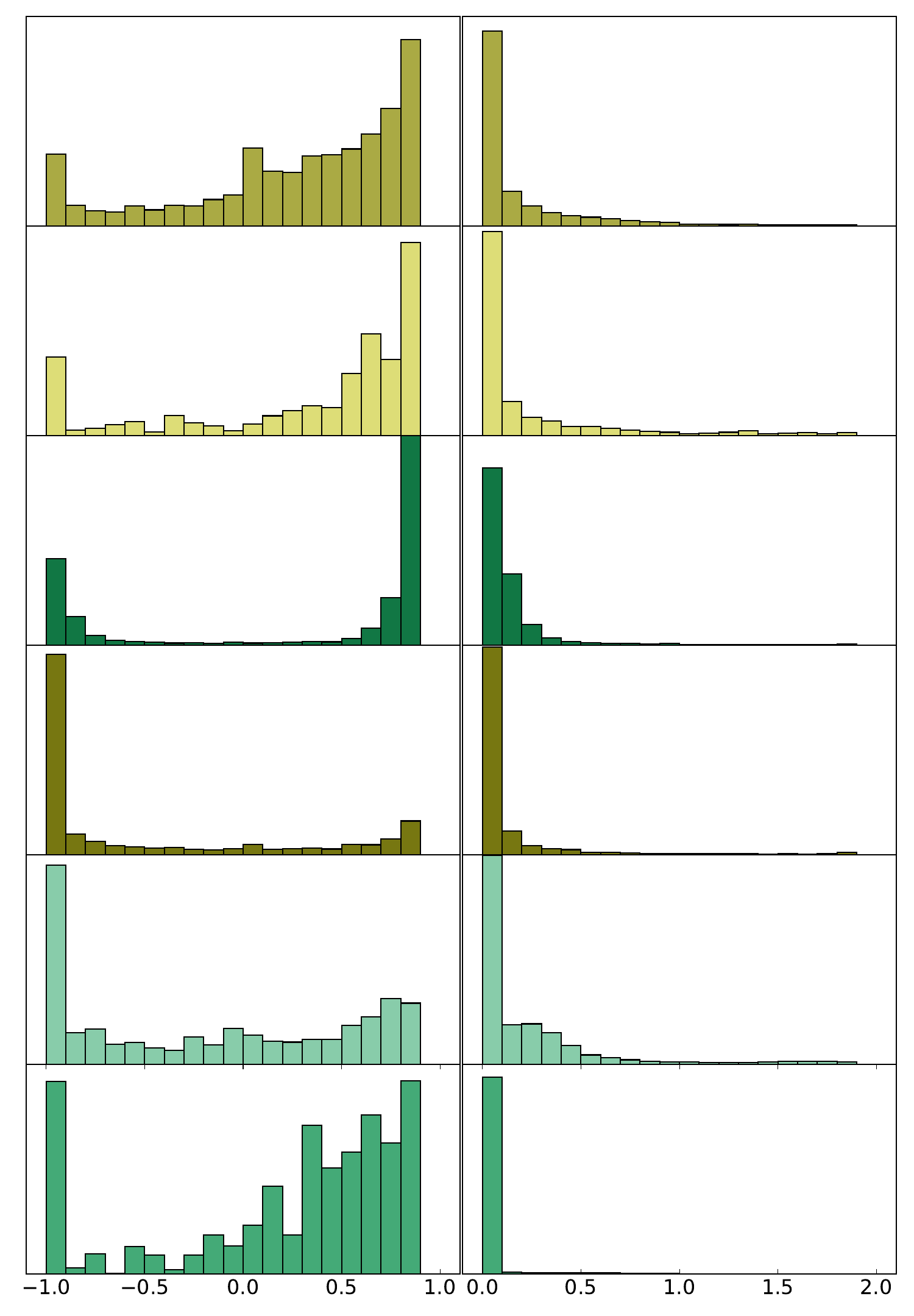}};
	\node[xshift=-.1\textwidth, align=center] (cL) at (bNIPS.south) {$\ell(v)$};
	\node[xshift=.1\textwidth, align=center] (cR) at (bNIPS.south) {$\abs{\ell(v)-\ell(u)}$};
	\node[xshift=-1, rotate=90] (lR) at (bNIPS.west) {};
	\node[above of=lR, align=center, rotate=90, node distance=5cm] (la) {\dtstFrack{}};
	\node[above of=lR, align=center, rotate=90, node distance=3.cm] (la) {\dtstBrexit{}};
	\node[above of=lR, align=center, rotate=90, node distance=1.cm] (la) {\dtstUselect{}};
	\node[below of=lR, align=center, rotate=90, node distance=1.cm] (la) {\dtstAbort{}};
	\node[below of=lR, align=center, rotate=90, node distance=3.cm] (la) {\dtstObamacare{}};
	\node[below of=lR, align=center, rotate=90, node distance=5cm] (la) {\dtstPhone{}};
	
	\node[anchor=south] (bDBLP) at ({-\x}, -.080) {\includegraphics[width=.45\textwidth]{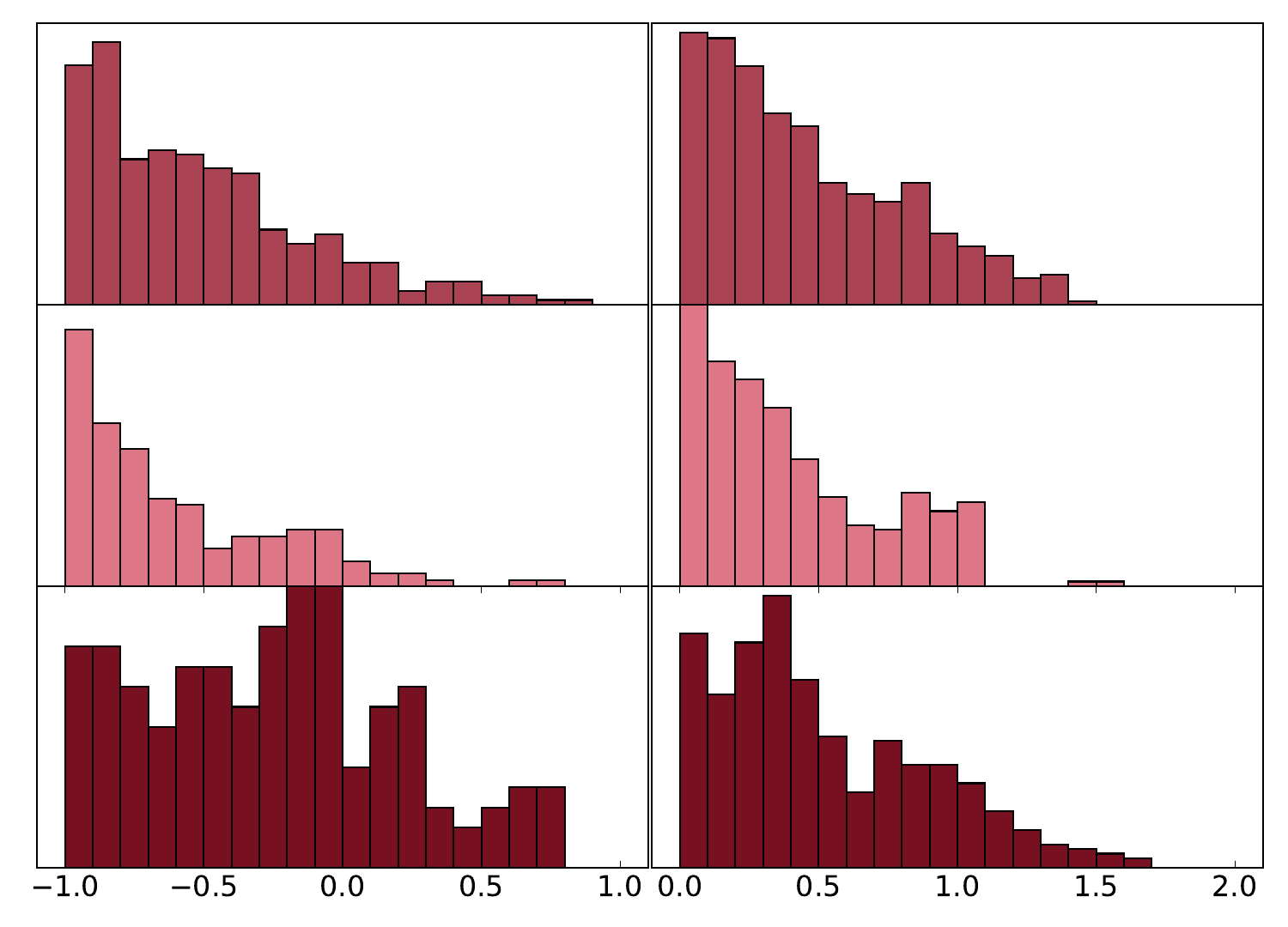}};
	\node[xshift=1, rotate=90] (lR) at (bDBLP.west) {\dtstDBLPbs{}};
	\node[above of=lR, rotate=90, node distance=1.8cm] (la) {\dtstDBLPpy{}};
	\node[below of=lR, rotate=90, node distance=1.8cm] (la) {\dtstDBLPcp{}};
	
	\node[anchor=north] (bUSEP) at ({-\x}, 0.080) {\includegraphics[width=.45\textwidth]{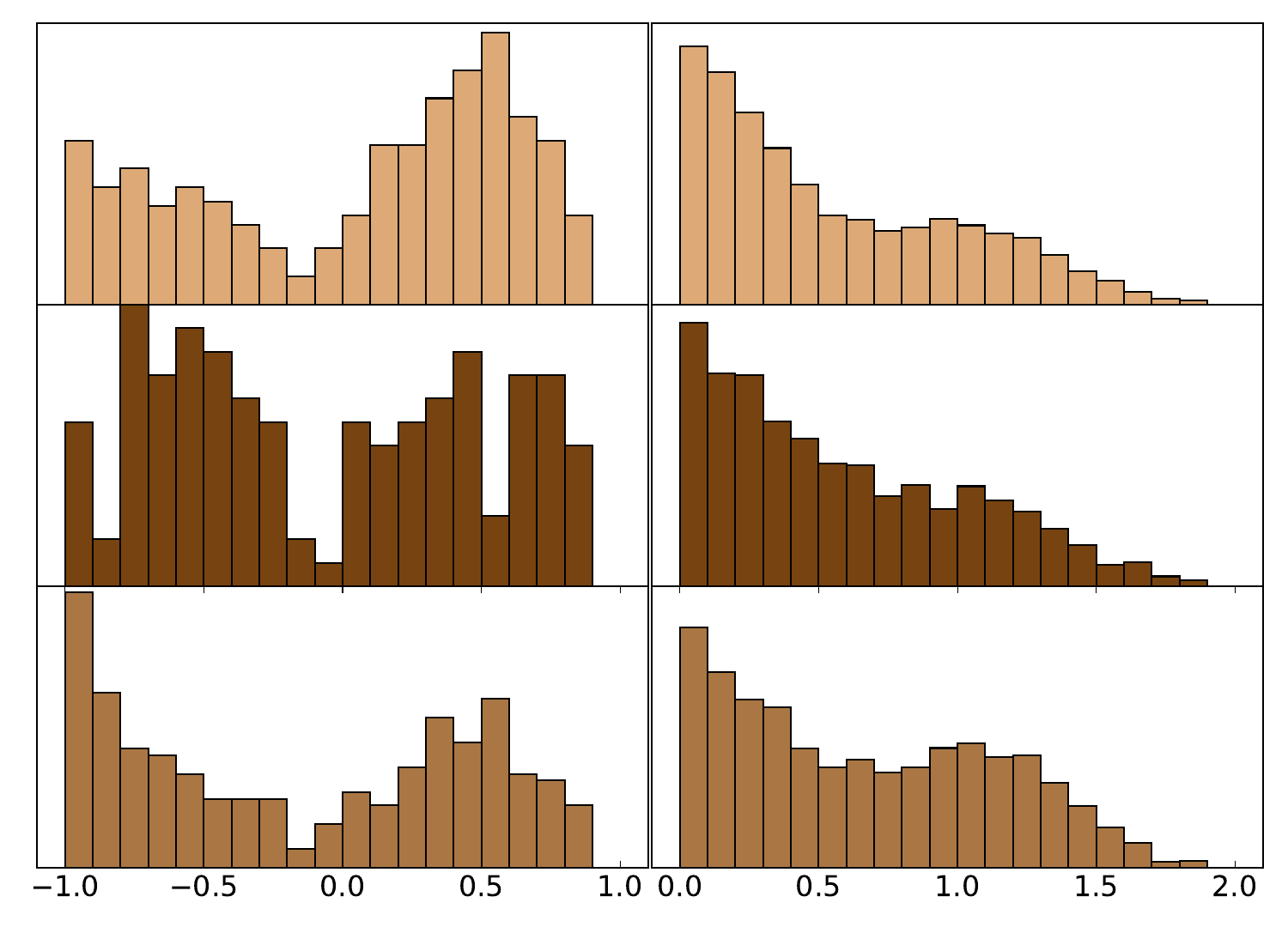}};
	\node[xshift=-.1\textwidth, align=center] (cL) at (bUSEP.south) {$\ell(v)$};
	\node[xshift=.1\textwidth, align=center] (cR) at (bUSEP.south) {$\abs{\ell(v)-\ell(u)}$};
	\node[xshift=1, rotate=90] (lR) at (bUSEP.west) {\dtstTWEPc{}}; 
	\node[above of=lR, rotate=90, node distance=1.8cm] (la) {\dtstTWEPa{}}; 
	\node[below of=lR, rotate=90, node distance=1.8cm] (la) {\dtstTWEPb{}}; 
	\end{tikzpicture}
	\caption{Histograms of node leanings (left) and leaning differences across the edges (right) of \dtstClDBLP{}, \dtstClTWEP{} and \dtstClNIPS{} networks.}
	\label{tab:app_histEP}
	\label{tab:app_histDBLP}
	\label{tab:app_histNIPS}
\end{figure*}

\begin{figure}[tbh]
	\centering
	\begin{tikzpicture}[node distance=1.1em]
	\node[anchor=north] (bUSTE) at (0,0) {\includegraphics[width=.45\textwidth]{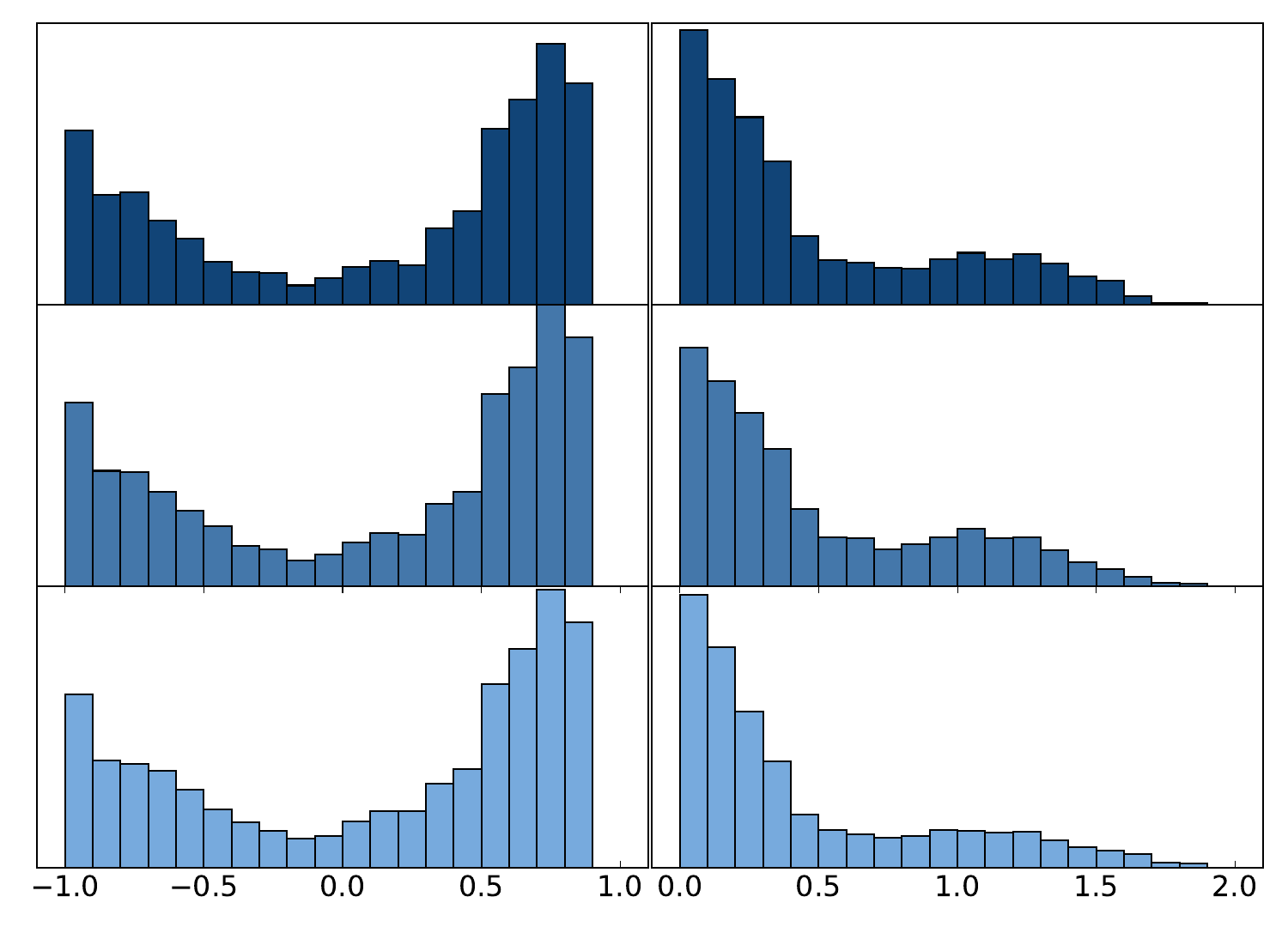}};
	\node[xshift=-.1\textwidth, align=center] (cL) at (bUSTE.south) {$\ell(v)$};
	\node[xshift=.1\textwidth, align=center] (cR) at (bUSTE.south) {$\abs{\ell(v)-\ell(u)}$};
	\node[xshift=1, rotate=90] (lR) at (bUSTE.west) {\dtstTWTESTwo{}};
	\node[above of=lR, rotate=90, node distance=1.8cm] (la) {\dtstTWTESFive{}};
	\node[below of=lR, rotate=90, node distance=1.8cm] (la) {\dtstTWTEMFive{}};
	\end{tikzpicture}
	\caption{Histograms of node leanings (left) and leaning differences across the edges (right) of \dtstClTWTE{} networks.}
	\label{tab:app_histTWTE}
\end{figure}

\begin{figure}[tbh]
\centering
\begin{tikzpicture}[node distance=1.1em]
\node[anchor=north, yshift=-4.15cm] (bTWIT) at (0,0) {\includegraphics[width=.45\textwidth]{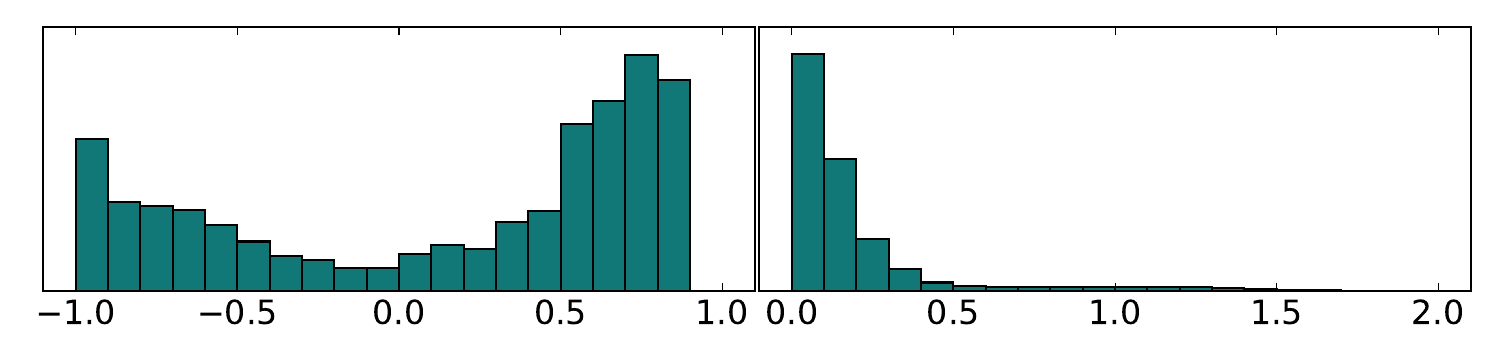}};
\node[xshift=-.1\textwidth, align=center] (cL) at (bTWIT.south) {$\ell(v)$};
\node[xshift=.1\textwidth, align=center] (cR) at (bTWIT.south) {$\abs{\ell(v)-\ell(u)}$};
\node[xshift=1, rotate=90] (lR) at (bTWIT.west) {\dtstTWTS{}};
\end{tikzpicture}
\caption{Histograms of node leanings (left) and leaning differences across the edges (right) of \dtstClTWT{} networks.}
\label{tab:app_histTWT}
\end{figure}

\subsection{Comparison baselines}

To better understand the quality of the returned assignments, 
we compare the solution of our algorithm to item--user assignments obtained by simple yet intuitive baselines. Recall that the running time of \algTdem{} is linear in the total number of generated RC-sets, which is very efficient. In order to not give it an unfair advantage against the comparison baselines, we store the RC-sets computed during the RC-set generation step of \algTdem{}, and use them to also compute the baselines.\!\footnote{Our implementation is publicly available: \url{https://github.com/aslayci/TDEM_extension}}

The first baseline, \algDClose{}, selects at each iteration the highest-degree node $v$ and greedily assigns this node the $k_v$ items that minimize the variance among the items assigned to $v$ so far. When $k_v$ items are assigned to $v$, \algDClose{} repeats the same steps for the next highest-degree node, until a total of $k$ assignments are obtained. The second baseline, \algDFar{}, operates almost identical to \algDClose, with the only difference being the maximization of the variance instead of minimization.  The third baseline, \algDWeighted{}, selects the (next) highest degree node $v$ at each iteration, like the other two baselines, but assigns a set $A_v \in H$ of $k_v$ items, in a greedy fashion, to maximize $f_v(A_v)$.
We also considered a simple baseline that uses fully random assignments. However it performed very poorly, obtaining exposure scores several orders of magnitude smaller than \algTdem{}, so we decided to leave it out.

\begin{table*}[tbp]
\centering
\caption{Results summary: Diversity exposure scores.}
\label{tab:results}
\begin{tabular}{@{\hspace{1ex}}l@{\hspace{6ex}}r@{\hspace{3ex}}r@{\hspace{3ex}}r@{\hspace{4ex}}r@{\hspace{5ex}}r@{\hspace{4ex}}r@{\hspace{1ex}}}
\toprule
  & \multicolumn{4}{c}{$\score$} & Mem.\ & RT \\
\cline{2-5} \\ [-.8em]
Dataset $(k, k_u)$ & \algDWeighted{} & \algDFar{} & \algDClose{} & \algTdem{} & (mb) & (s) \\
\midrule
\dtstDBLPbs{} $(5, 1)$ & \calcnum{22.75/(4*167)} & \calcnum{9.53/(4*167)} & \calcnum{28.01/(4*167)} & \calcnum{33.08/(4*167)} & $457$ & $2.16$ \\
\dtstDBLPcp{} $(5, 1)$ & \calcnum{44.38/(4*144)} & \calcnum{11.19/(4*144)} & \calcnum{40.21/(4*144)} & \calcnum{63.66/(4*144)} & $276$ & $1.8$ \\
\dtstDBLPpy{} $(5, 1)$ & \calcnum{133.85/(4*342)} & \calcnum{25.04/(4*342)} & \calcnum{176.68/(4*342)} & \calcnum{228.08/(4*342)} & $285$ & $2.68$ \\
\dtstTWEPc{} $(5, 1)$ & \calcnum{49.62/(4*140)} & \calcnum{14.56/(4*140)} & \calcnum{40.01/(4*140)} & \calcnum{72.27/(4*140)} & $279$ & $1.97$ \\
\dtstTWEPb{} $(5, 1)$ & \calcnum{236.64/(4*338)} & \calcnum{61.73/(4*338)} & \calcnum{210.25/(4*338)} & \calcnum{262.81/(4*338)} & $174$ & $2.34$ \\
\dtstTWEPa{} $(5, 1)$ & \calcnum{1036.91/(4*577)} & \calcnum{755.21/(4*577)} & \calcnum{810.99/(4*577)} & \calcnum{999.32/(4*577)} & $1\,658$ & $42.92$ \\
\dtstTWTESFive{} $(50, 5)$ & \calcnum{204.09/(4*2719)} & \calcnum{123.14/(4*2719)} & \calcnum{251.70/(4*2719)} & \calcnum{324.51/(4*2719)} & $5\,943$ & $24.12$ \\
\dtstTWTESTwo{} $(50, 5)$ & \calcnum{1531.78/(4*4379)} & \calcnum{375.05/(4*4379)} & \calcnum{2668.32/(4*4379)} & \calcnum{3096.67/(4*4379)} & $656$ & $9.41$ \\
\dtstTWTEMFive{} $(50, 5)$ & \calcnum{3877.27/(4*5183)} & \calcnum{1068.39/(4*5183)} & \calcnum{5297.45/(4*5183)} & \calcnum{6914.75/(4*5183)} & $3\,100$ & $77.47$ \\
\dtstTWTS{} $(50, 5)$ & \calcnum{4408.95/(4*5454)} & \calcnum{4349.02/(4*5454)} & \calcnum{2034.50/(4*5454)} & \calcnum{7041.20/(4*5454)} & $373$ & $44.07$ \\
\dtstBrexit{} $(50, 5)$ & \calcnum{135.72/(4*22745)} & \calcnum{85.02/(4*22745)} & \calcnum{119.36/(4*22745)} & \calcnum{268.14/(4*22745)} & $23\,725$ & $72.49$ \\
\dtstPhone{} $(50, 5)$ & \calcnum{3689.52/(4*36742)} & \calcnum{1032.30/(4*36742)} & \calcnum{2391.89/(4*36742)} & \calcnum{6658.25/(4*36742)} & $1\,803$ & $15.49$ \\
\dtstUselect{} $(50, 5)$ & \calcnum{106.90/(4*23816)} & \calcnum{311.758/(4*23816)} & \calcnum{380.14/(4*23816)} & \calcnum{844.44/(4*23816)} & $45\,828$ & $525.21$ \\
\dtstAbort{} $(50, 5)$ & \calcnum{640.06/(4*279505)} & \calcnum{69.55/(4*279505)} & \calcnum{130.09/(4*279505)} & \calcnum{902.10/(4*279505)} & $154\,588$ & $1\,275.25$ \\
\dtstFrack{} $(50, 5)$ & \calcnum{140.27/(4*374403)} & \calcnum{84.91/(4*374403)} & \calcnum{128.52/(4*374403)} & \calcnum{1143.86/(4*374403)} & $400\,565$ & $4\,785.12$ \\
\dtstObamacare{} $(50, 5)$ & \calcnum{130.95/(4*334617)} & \calcnum{77.97/(4*334617)} & \calcnum{134.02/(4*334617)} & \calcnum{1105.08/(4*334617)} & $360\,449$ & $3\,936.16$ \\
\dtstTWTXL{} $(50, 5)$ & \calcnum{98550.75/(4*481523)} & \calcnum{90019.75/(4*481523)} & \calcnum{65320.25/(4*481523)} & \calcnum{234803.75/(4*481523)} & $3\,438$ & $806.05$ \\
\bottomrule
\end{tabular}
\end{table*}

\subsection{Results}
Table~\ref{tab:results} shows the diversity exposure scores achieved by the three baselines and by our algorithm, \algTdem{}. For easier comparison, we report the average diversity exposure score of the individuals of the social network in each dataset. Recall that the smallest possible value is 0 and the maximum possible value is 1. Additionally, we report \algTdem{}'s memory consumption (in megabytes) and runtime (in seconds). The main computational bottleneck comes from the RC-generation step, which is also used by the baselines. Therefore, we do not report their memory consumption and runtime, since it differs only by a negligible amount to that of Greedy, as the rest of the computations performed by the baselines are trivial.
In summary, \algTdem{} clearly outperforms the simple baselines in terms of the diversity exposure scores obtained. \algTdem{} is able to identify non-trivial assignments that yield optimized diversity exposure in the network. That is, it finds a balance between exposure to diverse opinions yet selects items and nodes that do not have overly extreme leanings so as not to hinder propagation.

Observe that the runtime does not grow in proportion to the size of the network. Instead, it depends on the ability of items to propagate through the network, which depends, in turn, on the particular network structure, distribution of leanings, and propagation probabilities. Indeed, according to Theorem~\ref{theorem:stopCondCorrectness}, the more limited the propagation of items, the more samples are needed to ensure adequate estimation of the spread. Thanks to the use of reverse exposure sets, we obtain a highly efficient algorithm, especially considering that we are dealing with $h$ different influence spread problems, one for each item. 

\section{Conclusions}
\label{sec:conclusions}
In this paper we present the first work tackling the problem of maximizing the diversity of exposure in an item-aware information propagation setting, taking a step towards breaking filter bubbles. Our problem formulation models many aspects of real-life social networks, resulting in a realistic model and a challenging computational problem. Despite the inherent difficulty of the problem, we are able to devise an algorithm that comes with an approximation guarantee, and is very scalable thanks to a novel extension of \emph{reverse-reachable sets}. Through experiments on real-world datasets, we show that our method performs well and scales to large datasets.

Our work opens avenues for future work. One interesting problem is to improve the approximation guarantee of our algorithm by investigating further properties of the matroid formulation. Second, it would be interesting to experiment with different diversity functions, as well as to extend our approach for more complex propagation models such as, in particular, temporal variants of the Independent-Cascade model, with propagation probabilities that change over time.

\IEEEdisplaynontitleabstractindextext

%
\IEEEpeerreviewmaketitle

\ifCLASSOPTIONcaptionsoff
  \newpage
\fi

{\small
\bibliographystyle{IEEEtran}
\bibliography{propagation}
}


\begin{IEEEbiographynophoto}
	{Antonis Matakos} is a PhD student in the Data Mining Group of the Computer Science Department at Aalto University. He obtained his MSc degree from the University of Ioannina, Greece. His research interests belong broadly to the area of algorithmic data mining, with specific focus on social network analysis, graph theory and web mining. His PhD Thesis focuses on proposing social media models and algorithms for social good. 
\end{IEEEbiographynophoto}

\begin{IEEEbiographynophoto}
{Cigdem Aslay} is an assistant professor in the Department of Computer Science at Aarhus University. Previously, she was a postdoctoral researcher at Aalto University and ISI Foundation. She received her PhD from the University of Pompeu Fabra in December 2016. During her PhD, she was hosted by Yahoo! Research in the Web Mining Group. Her work focuses on algorithmic methods for graph mining and social network analysis.
\end{IEEEbiographynophoto}

\begin{IEEEbiographynophoto}
{Esther Galbrun} is a researcher at the School of Computing, University of Eastern Finland, working mostly on data mining methods. She is also a research scientist at Inria, France, currently on leave. She received her PhD in 2013 from the University of Helsinki. In 2018 she was a postdoctoral researcher in the Data Mining Group of the Computer Science Department at Aalto University, where the work for this paper was done.
\end{IEEEbiographynophoto}

\begin{IEEEbiographynophoto}
{Aristides Gionis}
	is a WASP professor in KTH Royal Institute of Technology, 
	an adjunct professor in Aalto University, and 
	a research fellow in ISI Foundation.
	Previously he was a senior research scientist in
	Yahoo!\ Research. He received his PhD from Stanford University in 2003. 
    He is currently serving as an associate editor in DMKD, TKDD, and TWEB.
	His research interests include data mining, web mining, and
	social-network analysis. 
\end{IEEEbiographynophoto}

\vfill
\balance
\end{document}


\maketitle
\label{sec:appendix}

\section{Proof of  Lemma 8}
\begin{proof}
For any given assignment $A$, let\\ $w_A = \xdivscore{A}/n$. Then, we have
\begin{align*}
& \prob\left(\abs[\Big]{n \, \mathcal{W}_{\sampleR}(A) - n \, w_A} \ge \frac{\epsilon}{2} \, \optvalue \right)  \\
&= \prob\left(\abs[\Big]{\theta \, \mathcal{W}_{\sampleR}(A) -\theta w_A} \ge  \, \frac{\theta \epsilon }{2n} \, \optvalue \right).
\end{align*}
By using Corollaries~1 and~2 and letting
$\delta = \dfrac{\epsilon\, \optvalue}{2 n w_A}$ we obtain
\begin{align*}
&\prob\left(\abs[\Big]{\theta \, \mathcal{W}_{\sampleR}(A) - \theta w_A} \ge \delta \theta w_A \right) \\
& \le 2 \exp\left(-\dfrac{\delta^2}{\frac{2\delta}{3} + 2} \, \theta w_A \right)\\
& = 2 \exp\left(-\dfrac{3 \epsilon^2 \, \optvalue^2}{4n ( \epsilon\,  \optvalue + 6 n w_A)}  \, \theta \right) \\
& \le 2 \exp\left(-\dfrac{3 \epsilon^2 \, \optvalue^2}{4n (\epsilon\,  \optvalue + 6 \, \optvalue)}  \, \theta  \right) \\
& = 2 \exp\left(-\dfrac{3 \epsilon^2 \, \optvalue}{4n (\epsilon + 6)}  \, \theta \right),
\end{align*}
where the last inequality above follows from the fact that
for any assignment of size at most $k$ we have $n w_A \le \optvalue$.
Finally, by requiring
\begin{align*}
2 \exp\left(-\dfrac{3 \epsilon^2 \, \optvalue}{4 n (\epsilon + 6)}  \, \theta \right)  \le \dfrac{1}{n^{\ell} \, \binom{n h}{k}},
\end{align*}
we obtain the lower bound on $\theta$.
\end{proof}

\begin{figure*}[tbp]
	\centering
	\begin{tikzpicture}[node distance=.66cm]
	\node[anchor=north,inner sep=0] (image) at (0,0) {\includegraphics[width=.8\textwidth]{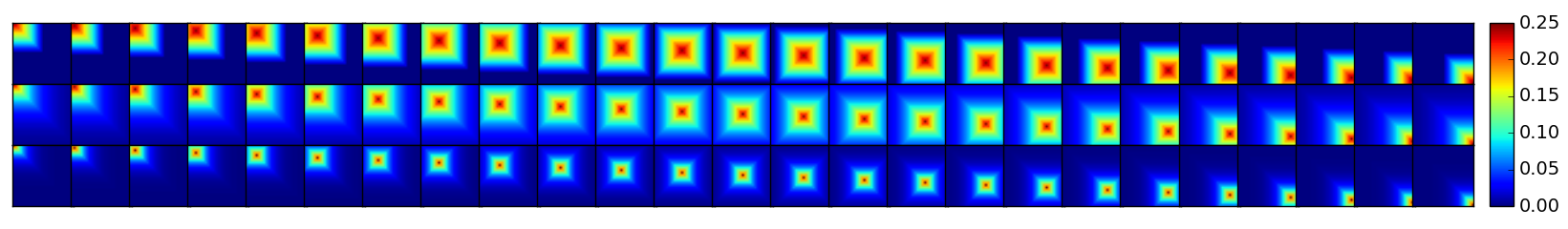}};
	\node (cW) at (image.west) {};
	\node[align=right, anchor=east, xshift=-14] (E2) at (cW) {$\pbMdlExpTwo$};
	\node[below of=E2] (E4) {$\pbMdlExpFour$};
	\node[above of=E2] (L) {$\pbMdlLin$};
	\node (cE) at (image.east) {};
	\node[above of=cE, xshift=-14, yshift=18] {$p^i_{uv}$};
	
	\node[yshift=27, xshift=-10, anchor=east, rotate=90]  (cl) at (cW) {$\ell(v)$};
	\node[yshift=22., xshift=-1, rotate=90, inner sep =1]  (ct) at (cW) {{\scriptsize -1}};
	\node[yshift=9.5, xshift=-1, rotate=90, inner sep =1]  (cb) at (cW) {{\scriptsize 1}};
	\draw[dotted] (ct) -- (cb);
	
	\node[yshift=38, xshift=0, anchor=west]  (kl) at (cW) {$\ell(u)$};
	\node[yshift=29, xshift=6, inner sep =1]  (kt) at (cW) {{\scriptsize -1}};
	\node[yshift=29, xshift=17, inner sep =1]  (kb) at (cW) {{\scriptsize 1}};
	\draw[dotted] (kt) -- (kb);
	
	\node[yshift=-38, xshift=186, anchor=west]  (il) at (cW) {$\ell(i)$};
	\node[yshift=-29, xshift=9, inner sep =1]  (it) at (cW) {{\scriptsize -1}};
	\node[yshift=-29, xshift=-32, inner sep =1]  (ib) at (cE) {{\scriptsize 1}};
	\draw[dotted] (it) -- (ib);
	\end{tikzpicture}
	\caption{Functions for computing the propagation probabilities.}
	\label{fig:app_probFs}
\end{figure*}

\section{Proof of  Lemma 9}
\begin{proof}
	Notice that the implication
	\[
	n \, \mathcal{W}_{\sampleR}(\approxgrA) \ge (1 + \epsilon) \, x \implies \optvalue \ge x
	\]
	is equivalent to
	\[
	\optvalue < x \implies n \, \mathcal{W}_{\sampleR}(\approxgrA) < (1 + \epsilon) \, x.
	\]
	Hence, to prove the lemma, we will show that when $\optvalue < x$,
	the probability that for any arbitrary assignment $A$ of $k$ pairs,
	$n \, \mathcal{W}_{\sampleR}(A) \ge (1 + \epsilon) \, x$ holds,
	is at most $\frac{n^{-\ell}}{\binom{n h}{k} \, \log_2 2n}$.
	Now, assume that $\optvalue < x$ and let $w_A = \xdivscore{A}/n$
	for an arbitrary assignment $A$ of $k$ pairs.
	We have
	\begin{align*}
		& \prob\left( n \, \mathcal{W}_{\sampleR}(A) \ge (1 + \epsilon) \, x \right) \\
		& = \prob\left(\theta \, \mathcal{W}_{\sampleR}(A) \ge \dfrac{(1 + \epsilon) \, x \, \theta}{n} \right) \\
		& = \prob\left(\theta \, \mathcal{W}_{\sampleR}(A) - \theta w_A \ge \dfrac{(1 + \epsilon) \, x \, \theta}{n} -\theta w_A \right) \\
		& = \prob\left(\theta \, \mathcal{W}_{\sampleR}(A) - \theta w_A \ge \left(\dfrac{(1 + \epsilon) \, x }{n \, w_A} -1\right) \, \theta w_A \right).
	\end{align*}
	
	Let $\delta = \left(\dfrac{(1 + \epsilon) \, x }{n \, w_A} -1\right)$.
	Notice that given $\optvalue < x$, we have
	\begin{align}\label{eq:lmLB1}
		w_A = \dfrac{\xdivscore{A}}{n} \le \dfrac{\optvalue}{n} < \dfrac{x}{n}.
	\end{align}
	Moreover, by Equation~(\ref{eq:lmLB1}), we have $\delta > \dfrac{\epsilon \, x}{n \, w_A} > \epsilon$.  Then, by Corollary~1 with $\delta = \left(\dfrac{(1 + \epsilon) \, x }{n \, w_A} -1\right)$, we obtain:
	\begin{align*}
		& \prob\left(\theta \, \mathcal{W}_{\sampleR}(A) - \theta w_A \ge \left(\dfrac{(1 + \epsilon) \, x }{n \, w_A} -1\right) \theta w_A \right) \\
		& \le \exp\left(-\dfrac{\delta^2}{\frac{2\delta}{3} + 2} \, \theta w_A \right) \\
		& = \exp\left(-\dfrac{\delta}{\frac{2}{3} + \frac{2}{\delta}} \, \theta w_A \right) \\
		& \le \exp\left(-\dfrac{\delta}{\frac{2}{3} + \frac{2}{\epsilon}} \, \theta w_A \right) \\
		& = \exp\left(-\dfrac{\epsilon \, x}{n \, w_A \, \left(\frac{2}{3} + \frac{2}{\epsilon} \right)} \, \theta w_A\right) \\
		& = \exp\left(-\dfrac{\epsilon^2}{\left(\frac{2\epsilon}{3} + 2 \right)} \, \dfrac{x}{n} \, \theta \right).
	\end{align*}
	Finally, given that
	\[
	\theta \ge \dfrac{(\frac{2}{3} \epsilon + 2) \, \left( \ln \binom{n h}{k} + \ell \ln n + \ln \log_2 n \right)}{\epsilon^2}  \, \frac{n}{x},\]
	we have
	\begin{align*}
		\exp\left(-\dfrac{\epsilon^2}{\left(\frac{2\epsilon}{3} + 2 \right)} \, \dfrac{x}{n} \, \theta \right) &\le \frac{n^{-\ell}}{\binom{n h}{k} \, \log_2 n}.
	\end{align*}
	
	To conclude, by the union bound,
	if $\optvalue < x$ then we have $n \, \mathcal{W}_{\sampleR}(\approxgrA) < (1 + \epsilon) \, x$
	with probability at least $1 - \frac{n^{-\ell}}{\log_2 n}$.
\end{proof}

\section{Proof of  Lemma 10}
\begin{proof}
	To prove this result we will show that when $\optvalue \ge x$, then
	the probability that \rcGreedy returns an assignment $A$ such that
	$n \, \Exp\left[\mathcal{W}_{\sampleR}(A)\right] > (1 + \epsilon) \,  \optvalue$
	is at most $\frac{n^{-\ell}}{\binom{n h}{k} \, \log_2 n}$.
	Then, by the union bound, if
	$\optvalue \ge x$, we have $n \, \mathcal{W}_{\sampleR}(\approxgrA)\le (1 + \epsilon) \, \optvalue$
	with probability at least $1 - \frac{n^{-\ell}}{\log_2 n}$.
	Recall that, as given by Lemma~9, we have
	\[
	\theta \ge \dfrac{(\frac{2\epsilon}{3} + 2) \, \left( \ln \binom{n h}{k} + \ell \ln n + \ln \log_2 n \right)}{\epsilon^2}  \, \frac{n}{x}.
	\]
	Then we have
	\begin{align*}
		& \prob\left(n \, \mathcal{W}_{\sampleR}(A) > (1 + \epsilon) \, \optvalue \right) \\
		& = \prob\left(\theta \, \mathcal{W}_{\sampleR}(A) >  (1 + \epsilon) \, \optvalue \, \frac{\theta}{n}  \right) \\
		& = \prob\left(\theta \, \mathcal{W}_{\sampleR}(A) - \theta w_A >  (1 + \epsilon) \, \optvalue \, \frac{\theta}{n} - \theta w_A \right) \\
		& \le \prob\left(\theta \, \mathcal{W}_{\sampleR}(A) - \theta w_A \ge  \frac{\epsilon \, \optvalue}{n \, w_A} \, \theta w_A  \right).
	\end{align*}
	
	Then, by Corollary~1 with $\delta =  \frac{\epsilon \, \optvalue}{n \, w_A}$, and using the fact that $\optvalue \ge n \, w_A$, for any $A$ of size $k$, we obtain:
	\begin{align*}
		& \prob\left(\theta \, \mathcal{W}_{\sampleR}(A) - \theta w_A \ge  \frac{\epsilon \, \optvalue}{n \, w_A} \, \theta w_A  \right) \\
		& \le \exp\left(-\dfrac{\delta^2}{\frac{2\delta}{3} + 2} \, \theta w_A \right) \\
		& = \exp\left(- \dfrac{\epsilon^2 \, \optvalue^2 \, 3n}{n^2 \, (6 n \, w_A +  2 \epsilon \, \optvalue)} \, \theta \right) \\
		& \le \exp\left(- \dfrac{3 \, \epsilon^2 \, \optvalue }{n \, (6 +  2\epsilon)} \, \theta \right) \\
		& \le \frac{n^{-\ell}}{\binom{n h}{k} \, \log_2 	n}.
	\end{align*}
\end{proof}

\section{Assignment of IC model probabilities}

{Heat map of the propagation probabilities: In Fig.~\ref{fig:app_probFs} we plot the values of three functions used to compute the probabilities $p^i_{uv}$ for the different combinations of leanings of the nodes and item, $\ell(u)$, $\ell(v)$ and $\ell(i)$.}

{More specifically, we considered a linear function with parameter $\beta$:
\[\Phi_{\text{lin}, \beta}(u, v, i) = \beta \cdot (1-\max(\abs{\ell(u) - \ell(i)}, \abs{\ell(v) - \ell(i)})/2)\;,\]
and an exponential function with parameters $\beta$ and $\gamma$:
\[\Phi_{\text{exp}, \beta, \gamma}(u, v, i) = \beta \cdot \exp(-\gamma \cdot \max(\abs{\ell(u) - \ell(i)}, \abs{\ell(v) - \ell(i)})/2)\;.\]}

{We set $\beta = 0.25$ and for the exponential functions $\gamma=2$ and $\gamma=4$. }

{We observe that for the exponential function with $\gamma=2$ we obtain the most reasonable values. The probability mass is spread more widely than for $\gamma=4$, which decays the probability too much, for diverging leanings. On the other hand, the linear function results in overall higher transmission probabilities. We feel that the exponential function for $\gamma=2$ best explains the filter bubble effect observed in social networks, while still allowing sufficient propagation of the items across the network. }